\documentclass[prd,reprint,superscriptaddress,floatfix,nofootinbib]{revtex4-2}

\usepackage{amsmath,amssymb,amsfonts,amsthm}
\usepackage{graphicx}
\usepackage{enumitem}
\usepackage{bm}
\usepackage{rotating}
\usepackage{hyperref}
\usepackage{pbox}
\usepackage{array}
\usepackage{mathtools}
\usepackage{leftidx}
\usepackage{lmodern}
\usepackage[normalem]{ulem}
\usepackage{braket}
\usepackage{tensor}
\usepackage[usenames,dvipsnames]{xcolor}
\usepackage{float}
\usepackage{footnote}
\usepackage{setspace}
\usepackage{CJKutf8}
\usepackage{dsfont}
\usepackage{mathrsfs}
\usepackage[many]{tcolorbox}
\allowdisplaybreaks

\newcommand{\mfd}{\mathcal{M}}
\newcommand{\hil}{\mathcal{H}}

\newcommand{\R}{\mathbb{R}}
\newcommand{\A}{\mathcal{A}}
\newcommand{\C}{\mathbb{C}}
\newcommand{\M}{\mathcal{M}}
\newcommand{\W}{\mathcal{W}}
\newcommand{\CS}{{C^\infty_0(\M)}}
\newcommand{\dd}{\mathrm{d}}
\newcommand{\rr}[1]{\left(#1\right)}
\newcommand{\bx}{{\bm{x}}}

\newcommand{\bk}{{\bm{k}}}
\newcommand{\kk}{|\bm{k}|}
\newcommand{\sx}{\mathsf{x}}
\newcommand{\sy}{\mathsf{y}}
\newcommand{\sz}{\mathsf{z}}

\newcommand{\ii}{\mathsf{i}}
\DeclareMathOperator{\supp}{\text{supp}}
\DeclareMathOperator{\tr}{\mathrm{Tr}}
\renewcommand{\tilde}{\widetilde}
\renewcommand{\bar}{\overline}
\newcommand{\AAA}{\text{A}}
\newcommand{\BB}{\text{B}}
\newcommand{\CC}{\text{C}}

\newcommand{\Sol}{\mathsf{Sol}}

\renewcommand{\Im}{\mathrm{Im}}

\newcommand{\ketbra}[2]{{\left| {#1} \right\rangle \!\!\left\langle {#2} \right|}}

\newcommand{\Fk}[1]{F_\bk^{(\text{#1})}}

\newcommand{\kako}[1]{\left( #1 \right)}

\newcommand{\ts}[1]{_{\mathrm{#1}}}

\newtheorem{proposition}{Proposition}
\newtheorem{assumption}{Assumption}

\begin{document}
 
\title{Bipartite and tripartite entanglement in pure dephasing relativistic spin-boson model}

\author{Kensuke Gallock-Yoshimura}
\email{gallockyoshimura@biom.t.u-tokyo.ac.jp} 
\affiliation{
Department of Information and Communication Engineering, 
Graduate School of Information Science and Technology,
The University of Tokyo, 
7–3–1 Hongo, Bunkyo-ku, Tokyo 113–8656, Japan
}
\affiliation{Department of Physics, Kyushu University, 
744 Motooka, Nishi-Ku, Fukuoka 819-0395, Japan}

\author{Erickson Tjoa}
\email{erickson.tjoa@mpq.mpg.de}
\affiliation{Max-Planck-Institut f\"ur Quantenoptik, Hans-Kopfermann-Stra\ss e 1, D-85748 Garching, Germany}

\date{\today}

\begin{abstract}

We study nonperturbatively the entanglement generation between two and three emitters in an exactly solvable relativistic variant of the spin-boson model, equivalent to the time-independent formulation of the Unruh-DeWitt detector model. We show that (i) (highly) entangled states of the two emitters require interactions very deep into the light cone, (ii) the mass of the field can generically improve the entanglement generation, (iii) while it is possible to find regimes with genuine Greenberger–Horne–Zeilinger-like tripartite entanglement, it is difficult find regimes where tripartite entanglement can be easily shown to be significant or classified. Result (iii), in particular, suggests that probing the multipartite entanglement of a relativistic quantum field nonperturbatively requires either different probe-based techniques or variants of the Unruh-DeWitt model. Along the way, we provide the regularity conditions for the $N$-emitter model to have well-defined ground states in the Fock space. 
\end{abstract}

\maketitle

\section{Introduction}

The Unruh-DeWitt (UDW) model \cite{Unruh1979evaporation,DeWitt1979} is a well-known model of light-matter interaction that has been used to study relativistic phenomena predicted from quantum field theory (QFT) in curved spacetimes, such as the Unruh and Hawking effects \cite{hawking1975particle,wald1994quantum,Crispino2008review,Aubry2014derivative,Unruh1979evaporation,DeWitt1979,tjoa2022unruh}. The model has been generalized in different ways and applied to many different contexts to study entanglement and communication (see, e.g., \cite{Lopp2021deloc,Tales2020GRQO,perche2022localized,Nadine2021delocharvesting,doukas2013unruh,hu2022qhorad,hotta2020duality,perche2023particle,gale2023relativisticCOM,tjoa2023qudit,perche2024udwfromQFT,tjoa2022capacity,Landulfo2016communication,Simidzija2020capacity,Simidzija2018no-go,kasprzak2024transmission,pologomez2024delta,gallock2025qudit,landulfo2024broadcast,tjoa2023nonperturbative,mendez2023tripartite}). In \cite{Tjoa2024interacting}, it was argued that it is sometimes useful to view the UDW model as a relativistic variant of the spin-boson (SB) model \cite{amann1991spinboson,spohn1989spinboson,fannes1988equilibrium}. Physically, the main difference between the standard UDW model and the relativistic spin-boson (RSB) model is that the latter is described by a time-independent Hamiltonian---as a closed system---while the former is time-dependent. Terminology-wise, a two-level system is often called a \textit{particle detector}\footnote{This is motivated by Unruh and Hawking effects, as particle number is observer-dependent \cite{Unruh1979evaporation,hawking1975particle,Takagi1986noise,birrell1984quantum,unruhwald1984detector}.} in the UDW model, and an \textit{emitter} in the (R)SB model; we will use them interchangeably.

Here we will first revisit the problem of entanglement generation between two UDW detectors interacting with the vacuum state of the noninteracting scalar field, and then initiate the study of tripartite entanglement generation in this model. 
We are particularly interested in nonperturbative results by going into the \textit{pure dephasing regime}, i.e., when the interaction Hamiltonian commutes with the detectors' free Hamiltonian \cite{tjoa2023nonperturbative}. 
The special case of this is when the detectors are \textit{gapless}, which has been studied in the UDW framework in the context of entanglement harvesting and relativistic communication \cite{Landulfo2016communication, Simidzija2018no-go, Simidzija2017coherent, landulfo2024broadcast, perche2024closedform}. 
We will, however, do this using the closed-system RSB formulation following \cite{Tjoa2024interacting} for several reasons: it gives us a basis for controlling the ultraviolet/infrared (UV/IR) regularity of the model, it makes it easier to interpret some of the results we present regarding the causal behavior of the model, and it follows more closely the standard practice in quantum optics. 
We note that while the dynamics of the pure dephasing regime and the gapless regime are the same, the kinematics are different: for example, the gapless RSB model has degenerate joint ground states \cite{Tjoa2024interacting,amann1991spinboson,spohn1989spinboson,fannes1988equilibrium}.

More specifically, we use the nonperturbative RSB model to study how the entanglement generated between two and three initially uncorrelated detectors via local interactions with the field depends on its wave properties dictated by the \textit{strong Huygens principle}, which governs the support of the Green's functions of the wave equation \cite{McLenaghan1974huygens, Sonego1992huygenscurved, Faraoni2019huygens}. 
There are primarily two reasons that motivate our analysis. 
First, on the one hand there is a no-go theorem forbidding entanglement extraction from the vacuum in the gapless model if the detectors' interaction regions are spacelike-separated \cite{Simidzija2018no-go}, while for a massless field in (3+1)D Minkowski spacetime two detectors can be entangled (as expected) if we allow for causal contact \cite{perche2024closedform}. 
On the other hand, perturbatively the maximal entanglement occurs at the light cone \cite{tjoa2021harvesting} as it is enhanced by signalling, which in turn depends on the strong Huygens principle. 
Second, we would like to study the tripartite entanglement scenario, which has been done for a closely related \textit{delta-coupled} variant \cite{mendez2023tripartite,mendez2022tripartite,membrere2023tripartite}, as this will inform us of the extent to which the gapless model is useful for studying multipartite entanglement in the context of QFT in curved spacetimes. 
These will in turn close the (remaining) gap in the literature for the gapless UDW model.

In this work, we show that (i) (highly) entangled states of the two emitters require interactions \textit{very deep} into the light cone---the \textit{entanglement cone} marking the boundary for nonzero entanglement is deeper in the light cone interior with increasing mass and spacetime dimension; (ii) the mass of the field can generically improve the entanglement generation, at the expense of longer interaction time; we argue, however, that this improvement is somewhat independent of the strong Huygens principle; (iii) while it is possible to find regimes with genuine GHZ-like tripartite entanglement, the parameter space for which this occurs is very small. 
Furthermore, it is difficult to find regimes where tripartite entanglement can be easily shown to be significant or can be easily classified into different inequivalent types. 
Result (iii), in particular, suggests that probing the multipartite entanglement of a relativistic quantum field nonperturbatively via localized probes requires either different probe-based techniques or variants of the UDW model.

This work is organized as follows. 
In Sec.~\ref{sec: setup}, we briefly sketch the setup and establish notation and state some UV/IR regularity requirements for the RSB model. 
In Sec.~\ref{sec: entanglement-two} we discuss the entanglement dynamics for two detectors/emitters. In Sec.~\ref{sec: entanglement-three} we discuss the entanglement dynamics for three detectors/emitters. We adopt the convention that $c=\hbar=1$ and use mostly plus signature for the metric. The symbol $\sx$ denotes a spacetime event.

\section{Setup}
\label{sec: setup}

{In this section, we briefly cover the basic ingredients involved in our analysis. We will first briefly review scalar field theory in the algebraic quantum field theory (AQFT), the formulation of UDW model as a closed-system RSB model, and the strong Huygens principle. In order to focus on the physics, we will only cover the bare minimum required to follow this work. The readers can find much more detailed mathematical construction in the following references: for AQFT, we mainly rely on a very accessible exposition of the scalar field theory in \cite{hack2015cosmological} (see also \cite{fewster2019algebraic,KayWald1991theorems,Khavkhine2015AQFT}), while for the UDW model we summarize the exposition in \cite{Tjoa2024interacting}, which is essentially inspired by \cite{spohn1989spinboson,fannes1988equilibrium}.}

\subsection{Scalar field theory in curved spacetime}

The starting point is the equation of motion for a real scalar field in an $(n+1)$-dimensional globally hyperbolic Lorentzian spacetime $(\M,g)$ where $g$ is the metric tensor. 
We have the Klein-Gordon equation
\begin{align}
    P\phi \coloneqq (-\square+m^2+\xi R)\phi(\sx) = 0\,,
    \label{eq: KGE}
\end{align}
where $\square = \nabla_a\nabla^a$ is the d'Alembert operator, $m$ is the mass parameter of the field, $R$ is the Ricci scalar curvature and $\xi\geq 0$ quantifies the degree of nonminimal coupling to gravity. 
Global hyperbolicity and the nature of the Klein-Gordon differential operator $P$ guarantee that the initial value problem for \eqref{eq: KGE} is well posed and the spacetime admits a foliation of the form $\M = \R\times \Sigma$, where $\Sigma$ is any Cauchy surface that serves as a constant-time slice. 
{For example, we can foliate Minkowski spacetime, the simplest spacetime with zero curvature, in terms of constant-$t$ spacelike hypersurfaces where $t$ is the usual inertial time coordinate.}

As a hyperbolic differential operator, $P$ admits unique \textit{advanced and retarded propagators} $E_A(\sx,\sx')=E_R(\sx',\sx)$: they are the Green's functions for $P$ where $E_R(\sx,\sy)$ [resp. $E_A(\sx,\sy)$] vanishes if $\sx$ is not contained in the causal future (past) of $\sy$. 
We can then define the so-called \textit{causal propagator}\footnote{In \cite{Tjoa2024interacting} it was clarified how the various conventions of the Green's functions lead to different conventions of the causal propagator, and they can be somewhat confusing. 
This is because these variations are not fixed solely by the metric signature. 
Fortunately, in most practical calculations, it does not really matter, as they only affect the overall minus sign that one can typically track. Here we follow \cite{hack2015cosmological}, which seems to be one of the cleanest.}
\begin{align}
    E(\sx,\sy)\coloneqq E_R(\sx,\sy)-E_A(\sx,\sy)\,,
\end{align}
which is now a kernel of $P$. The causal propagator is important as it appears in the {covariant} canonical commutation relation (CCR) of the quantized scalar field
\begin{align}
    [\hat{\phi}(\sx),\hat{\phi}(\sy)] = \ii E(\sx,\sy)\openone\,,
\end{align}
where $\openone$ is the identity operator. 
$E(\sx,\sy)$ is a (bi-)distribution since $\hat{\phi}(\sx)$ is an \textit{operator-valued distribution}, i.e., it is only a valid operator acting on some Hilbert space after smearing with some appropriate test function $f$. The mathematical formulation based on AQFT thus defines the \textit{smeared field operator}
\begin{align}
    \hat{\phi}(f)\coloneqq \int_\M \dd V\,f(\sx)\hat{\phi}(\sx)\,,
\end{align}
where $\dd V=\sqrt{-\det g}\, \dd^n \bx$ is the invariant volume element and $f\in \CS$ is a smooth compactly supported function on $\M$. Sometimes the spacetime smearing function $f$ is interpreted as a localization region of the operator, which reflects our inability to resolve a single point in spacetime.

Within the AQFT framework, the quantized scalar field theory is specified by an \textit{algebra of observables} $\A(\M)$ and a \textit{state} $\omega$. 
The algebra of observables $\A(\M)$ is formally defined as a unital $*$-algebra of all possible sums and products of smeared field operators $\hat{\phi}(f)$ for all possible test functions $f$ and their adjoints satisfying the smeared CCR
\begin{align}
    [\hat\phi(f),\hat{\phi}(g)] = \ii E(f,g)\openone\,,
    \quad 
    \forall f,g \in \CS
\end{align}
where
\begin{align}
    E(f,g) = \int_{\M \times \M} \dd V\dd V'\,f(\sx)E(\sx,\sx') g(\sx')\,.
\end{align}
The Klein-Gordon equation is enforced by the requirement that $\hat{\phi}(f) = 0$ whenever $f=Ph$ for some $h\in \CS$. Note that the properties of the Green functions guarantee that $E(f,g)=0$ whenever $f,g$ have spacelike supports. 
For real scalar fields, this imposes that the generators of $\A(\M)$ are smeared field operators with real test functions.

{In order to do physics, we need to specify states in addition to observables. A state is a $\C$-linear functional $\omega:\A(\M)\to \C$ satisfying positivity $\omega(A^\dagger A)\geq 0$ and normalization $\omega(\openone) = 1$ for $A\in \A(\M)$, i.e., it is a map that takes in an observable $A$ as an input and outputs its expectation value $\omega(A)$. 
Through the Gelfand-Naimark-Segal (GNS) reconstruction theorem, we can obtain the Hilbert space representation of the theory, with the GNS vector $\ket{\Omega_\omega}$ identified as the ``vacuum state'' of the theory, and expectation values are given by $\omega(A) = \braket{\Omega_\omega|\pi_\omega(A)|\Omega_\omega}$. 
Here $\pi_\omega: \A(\M)\to L(\mathcal{H}_\omega)$ is a GNS representation such that $A$ is represented as a linear operator $\pi_\omega(A) \in L(\hil_\omega)$ acting on a (dense subset of) the GNS Hilbert space $\mathcal{H}_\omega$. 
The main advantage of the AQFT framework is that it treats all Hilbert space representations democratically, as there are infinitely many unitarily inequivalent representations of the CCR algebra \cite{KayWald1991theorems, Khavkhine2015AQFT, fewster2019algebraic, hack2015cosmological}, a feature that is important for studying physical systems that are directly defined in the thermodynamic limit \cite{bratteli2002operatorv2}. 
Indeed, the spin-boson model has been studied in the algebraic framework for the same reason \cite{fannes1988equilibrium,spohn1989spinboson,amann1991spinboson,hasler2011ground,hasler2021existence}.

For physical applications, not all states are physically relevant. 
First, in a relativistic context, we need the states to be \textit{Hadamard states} \cite{Radzikowski1996microlocal, fewster2013hadamard, fewster2025hadamard}. 
Concretely, in $(3+1)$-dimensional spacetimes, this corresponds to the states whose Wightman two-point functions have the correct ``short-distance (UV) behavior'' and they take the form 
\begin{align}
    \mathsf{W}(\sx,\sy) \sim \frac{u(\sx,\sy)}{\sigma(\sx,\sy)} + v(\sx,\sy) \log(\sigma(\sx,\sy)) + w(\sx,\sy)\,, \label{eq:Hadamard Wightman}
\end{align}
where $u,v,w$ are smooth functions, and $\sigma(\sx,\sy)$ is \textit{Synge's world function}, defined as half the squared geodesic distance between $\sx,\sy \in \M$ with respect to the spacetime metric $g$. 
Eq.~\eqref{eq:Hadamard Wightman} should be viewed as a (bi-)distribution, where the first two terms in \eqref{eq:Hadamard Wightman} only depend on the metric and Klein-Gordon equation, and the state dependence only influences the UV-finite, regular part $w$.

Second, we need the GNS Hilbert space $\mathcal{H}_\omega$ to be separable: this is \textit{a priori} not guaranteed (see, e.g., \cite{witten2022doesQFTmake} for an explicit example in quantum spin chains in the thermodynamic limit). 
If we restrict, however, to a subclass of Hadamard states that are \textit{quasifree} (Gaussian), i.e., those whose Wightman $n$-point functions are fully specified by their 1- and 2-point functions via Wick's theorem, then the resulting quasifree representations give rise to Fock space $\mathcal{H}_\omega = \mathfrak{F}(\mathcal{H})$ (which is separable), where $\mathcal{H}$ \cite{fewster2019algebraic, Khavkhine2015AQFT, KayWald1991theorems, hack2015cosmological} is the so-called \textit{one-particle Hilbert space} associated with the positive-frequency solutions $\varphi_1, \varphi_2$ of the Klein-Gordon equation with Klein-Gordon inner product
\begin{align}
    \braket{\varphi_1,\varphi_2}_\textsf{KG} 
    = 
        \ii \int_{{\Sigma_t}} \dd \Sigma^a\,\varphi_1\nabla_a\varphi_2^* - \varphi_2\nabla_a\varphi_1^*\,.
\end{align}
Following this construction outlined in \cite{KayWald1991theorems,fewster2019algebraic, Khavkhine2015AQFT, Tjoa2024interacting, tjoa2023nonperturbative, kasprzak2024transmission}, we recover the standard Fock space construction in canonical quantization, where we can write the field operator in terms of the Fourier mode decomposition
\begin{align}
    \hat{\phi}(\sx) 
    = 
        \int_{\R^n} \dd^n\bk\,
        \hat{a}_\bk u_\bk(\sx) + \hat{a}_\bk^\dagger u^*_\bk(\sx)\,,
\end{align}
where $u_\bk(\sx)$ (resp. $u_\bk^*$) are delta-normalized eigenbasis of the one-particle Hilbert space $\mathcal{H}$ (resp. conjugate Hilbert space $\bar{\mathcal{H}}$)\footnote{The one-particle Hilbert space is related to the (complexified) solution space of \eqref{eq: KGE} $\Sol_\C(\M)$ via $\Sol_\C(\M) = \mathcal{H}\oplus \bar{\mathcal{H}}$, i.e., all complex solutions are linear combinations of positive and negative frequency modes.} called the \textit{positive (resp. negative) frequency modes}. 
These modes satisfy the Klein-Gordon equation \eqref{eq: KGE} and they are delta-normalized with respect to the Klein-Gordon inner product
\begin{align}
    \braket{u_\bk,u_{\bk'}}_\textsf{KG} 
    &= 
        \delta^n(\bk-\bk')\,, \notag\\
    \braket{u^*_\bk,u^*_{\bk'}}_\textsf{KG} 
    &= 
        -\delta^n(\bk-\bk')\,,\\ 
    \braket{u_\bk,u^*_{\bk'}}_\textsf{KG} 
    &=
        0\notag\,,
\end{align}
and the creation and annihilation operators obey the CCR $[\hat{a}_\bk,\hat{a}_{\bk'}^\dagger] = \delta^n(\bk-\bk')\openone$. 
The Fock vacuum state of the theory, associated with the GNS vector, is conventionally denoted as $\ket{0}$ and satisfies the distributional equality $\hat{a}_\bk\ket{0}=0$ for all $\bk$. 
All other states with finite particle number can then be built from applications of (smeared) creation operators on the vacuum state and taking linear combinations.

For many practical computations and in order to avoid domain issues associated with unbounded operators $\hat{\phi}(f)$, it is often convenient to work instead with the so-called \textit{Weyl algebra} $\W(\M)$, which is formally an exponentiated version of the field operator, i.e., it is a $C^*$-algebra generated by objects of the form $e^{\ii \phi(f)}$ (see \cite{fewster2019algebraic} for more details). 
One can then find the corresponding state whose GNS construction gives the same physical theory, with the GNS representation $\Pi_\omega:\W(\M)\to \mathcal{B}(\mathcal{H}_\omega)$ 
related via derivatives.\footnote{Strictly speaking, the states used for $\A(\M)$ and $\W(\M)$ are not the same, but there is a bijection between them \cite{ruep2021weakly}.} We note that the state $\omega$ is completely specified by its Wightman $n$-point function:
\begin{align}
    W(f_1, \ldots, f_n) 
    \coloneqq 
        \omega(\hat{\phi}(f_1) \ldots \hat{\phi}(f_n))\,.
\end{align}

For our purposes, it is useful to write
\begin{align}
    \hat{\phi}(f) 
    &= 
        \hat{a}(f_\bk^*) + \hat{a}^\dag(f_\bk) 
    \equiv 
        \hat{\phi}(f_\bk)\,,
\end{align}
where
\begin{align}
    \hat{a}(f_\bk^*) 
    = 
        \int_{\R^n} \dd^n \bk\,
        \hat{a}_\bk f^*_\bk\,,
        \quad 
    f_\bk 
    = 
        \int_\M \dd V\,f(\sx)u_\bk^*(\sx)\,.
\end{align}
This relabeling is useful as it emphasizes which of the field modes we are accessing by choosing a particular spacetime smearing function $f$. 
This ``momentum space'' representation of the field operators will be used often for the rest of this work. 
In Minkowski space, we note that the mode functions are exactly known, and thus they serve as examples where calculations can be done explicitly: in the standard inertial coordinates $(t,\bx)$, we have
\begin{align}
    u_\bk(t,\bx) = \frac{1}{\sqrt{2(2\pi)^n\omega_\bk}}e^{-\ii\omega_\bk t +\ii \bk\cdot \bx}\,,
\end{align}
where $\omega_\bk ^2= |\bk|^2+m^2$ is the relativistic dispersion relation. 
}

\subsection{UDW model as a time-independent RSB model}

The setup we have in mind involves two UDW detectors or emitters A and B interacting locally with a relativistic scalar field. 
Let us recall the standard prescription based on the covariant formulation of the original UDW model in terms of the Hamiltonian density in the interaction picture \cite{Tales2020GRQO}.

In the original UDW framework, we first assume that the center of mass (COM) of each detector follows a timelike trajectory $\sz_j(\tau_j)$, $j \in \{ \AAA,\BB \}$, parametrized by their own proper time $\tau_j$, and then in the COM frame we establish the so-called Fermi normal coordinates (FNC) adapted to the trajectory $\bar{\sx}\equiv (\tau,\bar{\bx})$ where the COM trajectory is parametrized by $\sz_j(\tau_j) \equiv (\tau_j,\bar{\bm{0}})$ \cite{poisson2011motion}. 
The point of setting up an FNC is to be able to make a physically reasonable assumption that the interaction between the detector and the field is specified by the spacetime smearing of the form $f(\sx(\bar{\sx})) = \chi(\tau)F(\bar{\bx})$, i.e., in the rest frame of the detector we can tell apart the spatial profile $F(\bar\bx)$ (associated with, say, atomic orbitals of an atom) from the \textit{switching function} $\chi(\tau)$ that determines the duration and strength of the interaction. 
This implicitly assumes that the spatial profile is time-independent within the rest frame, i.e., we are working in a regime where the detector is (Born-)rigid. 
The interaction Hamiltonian density in the interaction picture thus takes the form
\begin{align}
    \hat{h}_{\text{I}}^{(j)}(\bar{\sx}) 
    &= 
        \lambda_jf^{(j)}(\tau,\bar{\bx})\hat{O}_j(\tau) \otimes \hat{\phi}(\sx(\tau,\bar{\bx}))\,,
    \label{eq: UDW}
\end{align}
where $j=A,B$, $\lambda_j \geq 0$ is the coupling constant and $\hat{O}_j$ is some detector observable. 
We can then compute the unitary operator generated by this Hamiltonian in the interaction picture, namely
\begin{align}
    \hat{U}\ts{I}
    = 
        \mathcal{T}_t
        \exp 
        \left[
            -\ii \int_\M \dd V\,
            \hat{h}\ts{I}^{(\AAA)}(\sx)+\hat{h}\ts{I}^{(\BB)} (\sx)
        \right] ,
\end{align}
where $\mathcal{T}_t$ is a time-ordering symbol with respect to some global time function. In most cases, one has to perform a perturbative Dyson series truncation, and even then it is mostly intractable except in special cases, such as (conformally) flat spacetimes and specific trajectories, unless one takes a specific limit where the unitary operator above can be calculated exactly---this happens for the so-called delta-coupling and pure-dephasing models \cite{tjoa2023nonperturbative, landulfo2024broadcast, Landulfo2016communication, pologomez2024delta, Simidzija2018no-go, tjoa2022capacity, Simidzija2020capacity, kasprzak2024transmission, mendez2023tripartite, gallock2025qudit}.

As argued in \cite{Tjoa2024interacting}, the UDW model has the nice property that its interaction can be made strictly local in spacetime, so that two detectors A and B can be said to be spacelike separated or not by checking the supports of their interaction profiles $f^{(j)}$.  This support region is often taken to represent the portion of spacetime operationally accessible to the corresponding observer, although the algebra of observables that the observer can effectively access can be strictly larger than the interaction region.\footnote{For example, a pointlike detector has interaction profile $f$ supported only along its timelike trajectory, but the accessible algebra of observables is prescribed by its \textit{timelike envelope}, the smallest causal diamond containing $\supp f$.} {Taking this difference into account, we can instead reframe the UDW model in the spirit of standard quantum optics, which goes by the name of \textit{spin-boson (SB) models}. 
Traditionally, the SB model is defined in flat spacetime and directly in the momentum representation: for a single two-level emitter, its Hamiltonian reads
\begin{align}
    \hat{H} 
    &= 
        \frac{\Omega}{2} \hat{\sigma}_z 
        + 
        \Delta \hat \sigma_x 
        + 
        \int_{\R^n} \dd^n\bk\,
        \omega_\bk \hat{a}^\dagger_\bk \hat{a}_\bk &\notag\\
    &\hspace{0.5cm}+ 
        \hat{\sigma}_x
        \otimes 
        \int_{\R^n} \dd^n\bk\, 
        \rr{
            \hat{a}_\bk F^*_\bk + \hat{a}_\bk^\dagger F_\bk
        }\,,
    \label{eq: SB}
\end{align}
where $\Omega$ is the energy gap, $\Delta$ measures the shift due to the detector's internal ``magnetic'' response, the third term is the field's free Hamiltonian, and $F_\bk$ is called the \textit{coupling function} that prescribes which modes the detector couples to and by how much. 
By construction, this model is \textit{time-independent} as the coupling functions are not dependent on time. 
We can think of this as a  ``closed system'' formulation of the UDW model where we model the interaction to be always on. 
Despite this constant-switching interaction, the spacetime locality of observers is still well defined in the sense that the support of the interaction is localized in space, while the observables accessible to each observer are localized in space and time. 
Reformulating the UDW model in terms of the constant-switching SB Hamiltonian in \eqref{eq: SB}, we have 
\begin{align}
    \hat H_{\Omega, \Delta, \lambda}
    &=
        \sum_j \hat H_0^{(j)}
        + 
        \hat H_{\phi,0}
        +
        \sum_j \hat H\ts{int}^{(j)}\,,\label{eq: SB-UDW}
\end{align}
where $\hat H_0^{(j)}$ and $\hat H_{\phi,0}$ are the free Hamiltonians for detector-$j$ and the field, respectively, and $\hat H\ts{int}^{(j)}$ is the interaction Hamiltonian that couples detector-$j$ and the field. 
These are explicitly written as 
\begin{subequations}
\begin{align}
    \hat H_0^{(j)}
    &=
        \frac{\Omega_j}{2}\hat{\sigma}_z^{(j)} 
        + \Delta_j \hat\sigma_x^{(j)}\,, \label{eq:RSB free Hamiltonian}\\
    \hat H_{\phi,0}
    &=
        \int_{\R^n} \dd^n\bk\,
        \omega_\bk\hat{a}^\dagger_\bk\hat{a}_\bk \\
    \hat H\ts{int}^{(j)}
    &=
        \hat{\sigma}_x^{(j)}
        \otimes
        \int_{\R^n} \dd^n\bk\,
        \rr{\hat{a}_\bk F_{\bk}^{(j)*} + \hat{a}_\bk^\dagger F_{\bk}^{(j)} }\,, \label{eq:RSB interaction Hamiltonian}
\end{align}
\end{subequations}
where 
\begin{align}
    F_{\bk}^{(j)} 
    &= 
        \lambda_j 
        \int_{\Sigma_\tau} \dd \Sigma_\tau\,
        F^{(j)}(\bar{\bx})u^*_\bk(0,\bx(\bar{\bx}))\,,
    \label{eq: coupling function}
\end{align}
with respect to some foliation associated with constant-$\tau$ surfaces. Note that the coupling constants are dimensionful as $[\lambda_j] = [\text{Length}]^{(n-3)/2}$ in $(n+1)$-dimensional spacetime. 

For convenience, we will distinguish the original UDW model described by the Hamiltonian \eqref{eq: UDW} from our time-independent, SB-type formulation in \eqref{eq: SB-UDW} by calling the latter the relativistic spin-boson (RSB) model. This terminology can be understood from the idea that in momentum space, what distinguishes the model from a generic SB model is the relativistic dispersion. 
}

There are obvious limitations to this formulation that the original UDW model \eqref{eq: UDW} does not have \cite{Tjoa2024interacting}. 
As currently stated, it only works for \textit{static trajectories} with respect to the quantization frame, as otherwise $F_{\bk}^{(j)}$ in Eq.~\eqref{eq: coupling function} would be time-dependent, which would correspond to time-dependent coupling (as is the case for Unruh effect or adiabatic switching). 
It follows that, for multiple detectors, all COM trajectories are static relative to one another, and must align with the Killing vector field generating time translations of the spacetime. 
Consequently, the SB framework with a time-independent Hamiltonian by default excludes some of the natural settings where the usual UDW approach thrives, e.g., the analysis of relative motion or accelerating trajectories. 
We will not discuss how to handle these issues in this work, and we refer to \cite{deBievre2006unruh} for how this is handled in the case of the Unruh effect.

{These limitations notwithstanding, we name two reasons why the RSB formulation is relevant for us.
\begin{enumerate}[label=(\alph*),leftmargin=*]
    \item It allows for a rigorous treatment of the model, i.e., whether it has good UV and IR properties. 
    For example, in the single-emitter case, the existence of ground states of the model depends on the UV/IR regularity of the coupling function $F_\bk$ \cite{fannes1988equilibrium, amann1991spinboson, spohn1989spinboson}. 
    In the UDW language, this corresponds to making an appropriate choice of spatial smearing $F(\bx)$, whose properties are tied to the Klein-Gordon equation and the spacetime dimension \cite{Tjoa2024interacting}.\footnote{In the UDW framework, it is often the case that Gaussian smearing is used for calculational convenience, but its noncompact support leads to mild and quantifiable degrees of causality violation in these models \cite{Causality2015Eduardo, Bruno2020time-ordering, pipo2021causality, maria2023causality}. 
    Provided the detectors are not pointlike, spatial smearings often grant good UV regularity; what remains is to check that the model is also IR regular.} 
    There have been a few nice developments providing rigorous studies of the thermalization of an Unruh-DeWitt detector in the weak-coupling regime \cite{deBievre2006unruh,passegger2024disjoint,fewster2016waiting}.

    \item It allows for analysis of relativistic causality in a manner closer to the standard quantum optics. 
    For example, it is much simpler to specify how long the interaction has to occur before the detectors can effectively signal to one another by checking that $t\gtrsim L$, where $L$ is the separation of their COM. 
    This is particularly useful for one of our results, where we will show that significant entanglement is only generated ``deep in the light cone interior'' (c.f., Sec.~\ref{sec: entanglement-two}).  
\end{enumerate}
For completeness, we will extend the regularity requirements for single detectors in \cite{Tjoa2024interacting} to $N$ detector settings. 
These will be relevant for a better understanding of how the entanglement dynamics of the two initially separable detectors depend on the mass of the field and the spacetime dimensions. 
This goes by the name \textit{strong Huygens principle}, which we briefly discuss below.
}

\subsection{Strong Huygens principle}

One of our goals is to understand the impact of the mass and spacetime dimensions on the entanglement dynamics between two detectors in the nonperturbative regime, since these modify the signaling contribution to the dynamics. 
The perturbative analysis has been covered in the literature \cite{tjoa2021harvesting, Casals2020commBH, hector2022optimize}.

Roughly speaking, a field $\phi$ is said to satisfy the strong Huygens principle when its Green's function has distributional support only on the light cone \cite{McLenaghan1974huygens}. Intuitively, this means that an impulsive light emanated from an emitter can be captured only by a lightlike separated receiver. When support extends inside the light cone, we say the strong Huygens principle is violated, and a signal can be captured by a timelike receiver. By construction, massive fields violate the strong Huygens principle as massive excitations travel at subluminal speed. The simplest example of the strong Huygens principle is the massless scalar field in $(3+1)$-dimensional Minkowski spacetime, where in the inertial coordinates the causal propagator is given by
\begin{align}
    E(\sx,\sx') &= \frac{\delta(t-t'+|\bx-\bx'|) - \delta(t-t'-|\bx-\bx'|)}{4\pi |\bx-\bx'|}
\end{align}
Clearly, the field obeys the strong Huygens principle by virtue of the fact that the Dirac delta function has no support except when $\Delta t = \pm |\Delta \bx|$, i.e., two points are null-separated. 
It is also known that in $(n+1)$-dimensional flat spacetime, the wave equation $\square \phi=0$ fails the strong Huygens principle for odd spacetime dimension ($n$ even) and when $n=1,2$. 
In even spacetime dimensions, the strong Huygens principle is satisfied (see \cite{tjoa2021harvesting, Causality2015Eduardo} for some explicit computations).

Remarkably, some generic results about the violation of the strong Huygens principle can be obtained for curved spacetimes \cite{Faraoni2019huygens, Sonego1992huygenscurved}. 
Some intuition can be gleaned from the Klein-Gordon equation in \eqref{eq: KGE}: if we consider a generic nonminimal coupling parameter $\xi>0$, then massless fields coupled to gravity acquire effective mass due to the Ricci scalar contribution. 
Indeed, for constant-curvature geometry, this effectively gives a mass parameter $m_{\text{eff}} \coloneqq \sqrt{\xi R}$, so the Klein-Gordon equation is equivalent to that of a massive field minimally coupled to gravity ($\xi=0$). 
The importance of the strong Huygens principle is that it affects the propagation of information: indeed, massless fields that violate the strong Huygens principle can signal to the \textit{interior} of the light cone as the Green's function between two timelike-separated events is nonvanishing. 
The violation of the strong Huygens principle has been shown, using perturbative techniques, to impact how two UDW detectors can be entangled when they are in causal contact \cite{tjoa2021harvesting,Casals2020commBH} and how they can be exploited to send timelike signals using massless fields \cite{Blasco2015Huygens,Blasco2016broadcast}.

\subsection{Regularity conditions}
\label{sec: regularity}

Before we proceed with the analysis of the bipartite entanglement dynamics, we first take a detour to establish some regularity requirements for the RSB model for $N$ emitters. 
This will justify in part the use of our model instead of the usual interaction picture, and also explain why it will be useful to work with a covariant IR-cutoff $m \ll 1$ (in an appropriate unit) for dealing with the massless case. 
Readers interested only in the physical results pertaining to the entanglement dynamics of the two detectors can skip to Sec.~\ref{sec: entanglement-two}. In order to keep the discussion fairly accessible, we will try to frame our results in rigorous terms but with a fairly physicist-angled explanation.

The first assumption we will make is the following:
\begin{assumption}[pure-dephasing and no zero mode]
    \label{assume: pure-dephasing}
    The RSB model is chosen to be \textit{pure-dephasing}, i.e., with energy gap $\Omega_j=0$, such that the interaction Hamiltonian \eqref{eq:RSB interaction Hamiltonian} commutes with the free qubit Hamiltonians $\hat H_0^{(j)}= \Delta_j \hat{\sigma}_x^{(j)}$, i.e.,
    \begin{align}
        [\hat{H}\ts{int}^{(i)}, \hat H_0^{(j)}] = 0\,,\qquad \forall i,j\in \{\AAA,\BB,\CC,\ldots\}\,.
    \end{align}
    Furthermore, we will assume that the background spacetime is such that the Hilbert space representation admits no zero modes.
\end{assumption}
By zero mode we mean the field mode that is not associated with the 1-particle sector of the Fock space, i.e., it behaves like a free particle in quantum mechanics \cite{Tjoa2019zeromode,EMM2014zeromode,KayWald1991theorems,folacci1987desitter} and hence does not admit a Fock representation. This complicates the choice of the vacuum state on which the full Hilbert space can be built (see, e.g., \cite{Tjoa2019zeromode,EMM2014zeromode} for simple examples on Einstein cylinder), so we will avoid them for simplicity as they can be dealt only case-by-case.

Although it is not useful for studying weakly coupling phenomenon such as the Unruh effect, the pure-dephasing regime in Assumption~\ref{assume: pure-dephasing} is important as they can still exhibit rich dynamics and they enable rigorous results. More technically, the pure-dephasing model can be viewed as a $C^*$-dynamical system, which roughly speaking means that the dynamics can be implemented by an automorphism $\alpha^{0}_t:\A \to \A$ where $\A$ is the \textit{joint algebra of observables} of the detector-field system, independently of the Hilbert space representations. This is to be contrasted with the nonzero gap case, where the corresponding automorphism can only be rigorously defined as a \textit{pseudodynamics} \cite{bratteli2002operatorv2}, i.e., one has to work with an implementation of $\alpha_t^\Omega: \mathcal{B}(\mathcal{H}_\omega)\to \mathcal{B}(\mathcal{H}_\omega)$ in a particular Hilbert space representation. Morally speaking, this is due to the fact that for $\Omega\neq 0$, there is no guarantee that $\alpha_t^\Omega$ will not ``bring us outside the Hilbert space'' and give rise to divergences. The best one can do is to do a ``small $\Omega$'' perturbative calculation about the pure-dephasing as the perturbation is bounded \cite{spohn1989spinboson}. Treating the interaction Hamiltonian as a perturbation to the free Hamiltonian is much harder due to unbounded nature of the field operator (see, e.g., \cite{hasler2011ground,hasler2021existence}).

Apart from the issue of mathematical rigor, the main physical reason for choosing Assumption~\ref{assume: pure-dephasing} is for the sake of having rigorous results even for the strong coupling regime where the coupling constants $\lambda_j$ for each detector can be chosen to be arbitrarily large. 
In the original UDW case, it was shown that if the interaction Hamiltonian of the UDW model \eqref{eq: UDW} has spacelike supports, then the reduced density matrix of the two detectors after interaction is separable \cite{Simidzija2018no-go}. The key, of course, is spacelike separation as we know that generic interacting quantum systems that are in causal contact will generate some entanglement, and the same gapless model was shown to generate entanglement when they are in causal contact \cite{perche2024closedform}.

As we will show in Sec.~\ref{sec: entanglement-two}, one of our contributions will be to show that not only does the RSB model satisfy the no-go theorem, but that the detectors can only get entangled \textit{very deep into the light cone}. 
In other words, even though we can in principle make the final state arbitrarily close to the maximally entangled state, it takes much longer than the light-crossing time between them. 
To set up the stage for these results, we need more assumptions along the lines used for the single detector case in \cite{Tjoa2024interacting}. 
These additional assumptions require us to first study the ground state properties of the full detector-field system.

The first result is to guarantee that the ground state energy is bounded from below, assuming the existence of a ground state from the GNS construction.\footnote{We are following the construction in  \cite{spohn1989spinboson}: there, one assumes first the existence of a ground state in the form of some vector $\ket{\psi}\in \mathcal{H}_\omega\cong \C^2\otimes \mathfrak{F}(\mathcal{H})$, and then find the ground state energy. The inadequacy of this construction comes from the issue of whether $\mathfrak{F}(\mathcal{H})$ is the \textit{noninteracting Fock space}, in which case one has to construct the state by, e.g., taking zero-temperature limit of some thermal state. From the calculational point of view, it will turn out that the ground state energy will be the same either way as we can just view $\ket{\psi}$ as a GNS vector.} 
\begin{proposition}
    \label{prop: ground-state energy}
    Consider $N$ detectors, $\mathrm{A}, \mathrm{B}, \ldots$, described by the RSB Hamiltonian \eqref{eq: SB-UDW} in $(n+1)$-dimensional spacetime. 
    Then, for gapless detectors $\Omega_j = 0$, the expectation value of the Hamiltonian  $\braket{\hat H_{0,\Delta, \lambda}}$ has a lower bound given by
    \begin{align}
        \braket{\hat H_{0, \Delta, \lambda}}
        &\geq 
        \bm s\cdot \bm \Delta 
        - 
        R_1(\bm s \cdot \bm{F_k},n)\,,
    \end{align}
    for some choice of $\bm s \in \{ \pm \}^N$, where $\bm \Delta = [ \Delta\ts{A}, \Delta\ts{B}, \ldots ]^\intercal$ and $\bm{F_k}=[ F_\bk^{(\mathrm{A})}, F_\bk^{(\mathrm{B})},\ldots ]^\intercal$ are the $N$-dimensional vectors dependent on $\Delta_j$ and the coupling functions, respectively, and 
    \begin{align}
        R_\alpha(f_\bk,n) 
        &\coloneqq 
            \int_{\R^n} \frac{\dd^n\bk}{\omega^\alpha_\bk} |f_\bk|^2\,,
        \quad 
            \alpha \in \R_0^+\,.
    \end{align}
    Hence, the RSB Hamiltonian is bounded from below if $R_1(\bm s \cdot \bm{F_k},n) < \infty$.
\end{proposition}

\begin{proof}
For the sake of generality, we consider $N$ detectors A, B, C, $\ldots$, coupled to the quantum scalar field $\hat \phi$. 
Let $\ket{\pm_j}$ be the eigenvectors of $\hat \sigma_x^{(j)}$, $j\in \{ \AAA, \BB, \CC,\ldots \}$ with the eigenvalues $\pm 1$, respectively. 
Using the expressions $\hat \sigma_x^{(j)}=\ketbra{+_j}{+_j} - \ketbra{-_j}{-_j}$ and $\openone_j = \ketbra{+_j}{+_j} + \ketbra{-_j}{-_j}$, the RSB Hamiltonian for gapless detectors can be expressed as follows: 
\begin{align}
    \hat H_{0, \Delta, \lambda}
    &=
        \sum_{\bm s \in \{ \pm \}^N}
        \ketbra{\bm s}{\bm s} \otimes \hat H_{\bm s}\,,
        \label{eq:diagonalized RSB Hamiltonian}
\end{align}
where $\bm s \in\{ +, - \}^N$ is a vector composed of signs, and 
\begin{align}
    \hat H_{\bm s}
    &\coloneqq
        \hat H_{\phi,0}
        + \bm s\cdot \bm \Delta \openone_\phi 
        + \hat a( \bm s \cdot \bm F^*)
        + \hat a^\dag( \bm s \cdot \bm F ) \,, \label{eq:Hpq}
\end{align}
where $\bm \Delta \equiv [ \Delta\ts{A}, \Delta\ts{B}, \cdots ]^\intercal $ and $\bm{F_k} \equiv [  F_\bk^{(\AAA)},  F_\bk^{(\BB)} , \ldots ]^\intercal $. 
As an example, consider two detectors A and B (i.e., when $N=2$). 
In this case, we have $\ket{\bm s} \in \{ \ket{+\ts{A},+\ts{B}}, \ket{+\ts{A},-\ts{B}}, \ket{-\ts{A},+\ts{B}}, \ket{-\ts{A},-\ts{B}} \}$, and so the gapless RSB Hamiltonian reads 
\begin{align}
    \hat H_{0, \Delta, \lambda}
    &=
        \sum_{p,q=\pm}
        \ketbra{p\ts{A}}{p\ts{A}} 
        \otimes 
        \ketbra{q\ts{B}}{q\ts{B}} 
        \otimes 
        \hat H_{p,q}\,, \notag 
\end{align}
where 
\begin{align}
    \hat H_{p,q}
    &=
        \hat H_{\phi,0}
        + (p \Delta\ts{A} + q\Delta\ts{B}) \openone_\phi 
        + \hat a( p  F_\bk^{(\text{A})*} + q F_\bk^{(\text{B})*}) \notag \\
        &
        + \hat a^\dag( p F_\bk^{(\text{A})} + q F_\bk^{(\text{B})} )\,. 
        \label{eq:bipartite diagonal Hamiltonians}
\end{align}

The operator $\hat H_{\bm s}$ can be deformed into 
\begin{align}
    &\hat H_{\bm s}\notag\\
    &=
        \int_{\R^n} \!\!\!\!\dd^n \bk 
        \kako{
            \omega_\bk \hat a_\bk^\dag \hat a_\bk
            + \bm s \cdot \bm{F_k}^*\, \hat a_\bk 
            + \bm s \cdot \bm{F_k}\, \hat a_\bk^\dag 
            + \dfrac{ |\bm s \cdot \bm{F_k}|^2 }{\omega_\bk} \openone_\phi
        } \notag \\
        &
         + \bm s \cdot \bm \Delta \openone_\phi
        - \int_{\R^n} \dd^n \bk\,\dfrac{ |\bm s \cdot \bm{F_k}|^2 }{\omega_\bk} \openone_\phi \notag  \\
    &=
        \hat H_{\phi,0,\bm s}
        + 
        [ 
            \bm s \cdot \bm \Delta 
            - 
            R_1(\bm s\cdot \bm{F_k},n) 
        ] 
        \openone_\phi \,,
\end{align}
where we have defined 
\begin{align}
    \hat H_{\phi,0,\bm s}
    &\coloneqq
        \int_{\R^n} \dd^n \bk\, 
        \omega_\bk \hat c_{\bk,\bm s}^\dag \hat c_{\bk,\bm s} \,, \\
    \hat{c}_{\bk,\bm s} 
    &\coloneqq 
        \hat a_\bk 
        + 
        \frac{ \bm s\cdot \bm{F_k} }{\omega_\bk} \openone_\phi \notag \\
    &=
        \hat D^\dag (\bm s\cdot \bm{F_k}/\omega_\bk)
        \hat a_\bk 
        \hat D(\bm s\cdot \bm{F_k}/\omega_\bk)\,.
\end{align}
Here, $\hat D(\alpha_\bk)$ is the displacement operator. 
One can straightforwardly show that the field's free Hamiltonian $\hat H_{\phi,0}$ is unitarily equivalent to $\hat H_{\phi,0,\bm s}$ via a transformation: 
\begin{align}
    \hat H_{\phi,0,\bm s}
    &=
        \hat D^\dag (\alpha_{\bm s})
        \hat H_{\phi,0} 
        \hat D (\alpha_{\bm s})\,, 
    \quad
    \alpha_{\bm s}
    \coloneqq
        \int_{\R^n} \dd^n \bk\,
        \dfrac{\bm s \cdot \bm{F_k}}{\omega_\bk}\,.
\end{align}
This suggests that the ground state of $\hat H_{\phi,0,\bm s}$ is given by the coherent state $\ket{\alpha_{\bm s}}\equiv \hat D^\dag(\alpha_{\bm s}) \ket{0}$, where $\ket{0}$ is the Fock vacuum of the free Hamiltonian $\hat H_{\phi,0}$. 
Therefore, defining $\ket{\psi_0^{(\bm s)}}\coloneqq \ket{\bm s} \otimes \ket{\alpha_{\bm s}}$, $\bm s\in \{ \pm \}^N$, we obtain 
\begin{align}
    \braket{\hat H_{0, \Delta, \lambda}}
    &=
        \bm s\cdot \bm \Delta 
        - 
        R_1(\bm s \cdot \bm{F_k},n) \,.
\end{align}
    
\end{proof}

Essentially, Proposition~\ref{prop: ground-state energy} can be understood as the fact as follows: by performing appropriate Bogoliubov transformation the joint interacting ground state is essentially a tensor product of the eigenstates of $\hat{\sigma}_x$ and some (multimode) coherent state of the field whose coherent amplitude depends on the detector-field coupling. The only subtlety is that this argument assumes that the total number of particles in the ground state is finite.

\begin{proposition}
    The joint interacting ground state has a finite particle number if $R_2(\bm s \cdot \bm{F_k},n)<\infty$ for all $\bm s$. 
\end{proposition}
\noindent The proof is a simple adaptation of \cite{spohn1989spinboson} except that we need to consider all choices of signs $\bm s$ and that the integral $R_2$ is with respect to the ``combined smearing function'' $\bm s\cdot \bm{F_k}$. 
In particular, we see that the model exhibits IR divergences (production of an infinite number of soft bosons in the thermodynamic limit) if at least one detector is not coupled to the field with a sufficiently IR-regular coupling function $F^{(j)}_\bk$.

As argued in \cite{Tjoa2024interacting}, this suggests that modeling a physical system with the RSB or UDW framework requires careful control of both UV and IR behavior. The UV regularity is automatically taken care of by giving each detector a finite size (effectively serving as a UV cutoff), but the IR case is more subtle. If the field is massive or gains effective mass through, say, coupling to the (classical) gravitational field, then the model is also IR-regular. In contrast, when there are soft bosons, such as (3+1)-dimensional Minkowski spacetime with massless field, the key insight to be learned from \cite{spohn1989spinboson,fannes1988equilibrium,amann1991spinboson,hasler2011ground,hasler2021existence} is to consider a different Hilbert space representation that is not unitarily equivalent to the vacuum representation: indeed, the Araki-Woods representation and the thermal KMS representation suffice for the situation at hand \cite{araki1963representations,morfa2012deformations} (see \cite{passegger2024disjoint} for very recent analysis on KMS states). At the level of correlation functions, the vacuum correlators can be computed by taking the zero-temperature limit of the thermal correlators.

Since in this work we are going to focus on entanglement dynamics, the situation involving the infinitely many soft boson productions will not be of concern to us. Indeed, perturbatively one can check (following the strategy in \cite{tjoa2021harvesting}) that very light scalar field ($m\approx 0$) leads to bipartite entanglement between the detectors that is indistinguishable compared to the massless case. Thus we are going to assume this:
\begin{assumption}
    The scalar field is assumed to have mass $m>0$. 
    Massless fields will be viewed as having a covariant IR cutoff given by a small but finite mass $m$. 
\end{assumption}
\noindent {In the context of our results, we will consider the field to be massless for parameter $m\sigma\ll 10^{-11}$ where $\sigma$ is the effective size of the detector.} Crucially, this implies that the RSB model is both UV- and IR-regular in the sense above, it follows that the noninteracting Hilbert space is unitarily equivalent to the joint interacting Hilbert space, i.e., the interacting Hilbert space is $\C^2\otimes \C^2\otimes \mathfrak{F}(\mathcal{H})$ where $\mathfrak{F}(\mathcal{H})$ is the Fock space of the free field theory. In the subsequent analyses we will make these assumptions for simplicity, noting that one can always pass to the correct Hilbert space representation when the IR regularity needs to be dropped \cite{Tjoa2024interacting,fannes1988equilibrium,spohn1989spinboson}.

\subsection{General expression of the reduced density matrix for an $N$-partite system}
In this section, we derive the joint density matrix for $N$ gapless detectors in $(n+1)$-dimensional spacetime. 
Suppose the initial state for the total system is 
\begin{align}
    \ket{\psi(0)}
    &=
        \underbrace{\ket{g\ts{A}} \ket{g\ts{B}} \ldots }_N
        \otimes
        \ket{0} \notag \\
    &=
        \sum_{\bm s \in \{ \pm \}^N}
        \dfrac{1}{2^{N/2}} 
        \ket{\bm s } \otimes \ket{0}\,, \label{eq:N qubit initial state}
\end{align}
where $\ket{g_j}$ is the ground state of $\hat \sigma_z^{(j)}$, and $\ket{0}$ is the Fock vacuum of the field's free Hamiltonian $\hat H_{\phi,0}$. 
The second equality is derived by using $\ket{g_j}=( \ket{+_j}+\ket{-_j} )/\sqrt{2}$.

The total system of $N$ gapless qubits and the field evolves under a unitary operator generated by the gapless RSB Hamiltonian: $\hat U(t)=e^{-\ii t\hat H_{0,\Delta,\lambda}}$. 
As we have seen in Eq.~\eqref{eq:diagonalized RSB Hamiltonian}, the gapless RSB Hamiltonian takes the diagonalized form in the basis $\{ \ket{\bm s}\,|\, \bm s \in \{ \pm \}^N \}$. 
Therefore, the unitary operator $\hat U(t)$ can also be diagonalized in this basis, which can be expressed as 
\begin{align}
    \hat U(t)
    &=
        \sum_{\bm s \in \{ \pm \}^N}
        \ketbra{\bm s}{\bm s}
        \otimes e^{-\ii t \hat H_{\bm s}}\,,
\end{align}
and the total state $\ket{\psi(t)}$ at time $t$ can be compactly written as 
\begin{align}
    \ket{\psi(t)}
    &=
        \hat U(t) \ket{\psi(0)} \notag \\
    &=
        \dfrac{1}{2^{N/2}}
        \sum_{\bm s \in \{ \pm \}^N} 
        \ket{\bm s }  
        \otimes e^{ -\ii t \hat H_{\bm s} } \ket{0}\,.
\end{align}
The reduced density matrix for the $N$ gapless qubits thus reads 
\begin{align}
    \rho(t)
    &=
        \tr_\phi[ \ketbra{\psi(t)}{\psi(t)} ] \notag \\
    &=
        \dfrac{1}{2^N}
        \sum_{\bm s, \bm r \in \{ \pm \}^N} 
        \braket{0|e^{ \ii t \hat H_{\bm s} } e^{ -\ii t \hat H_{\bm r} }|0} \ket{\bm r}\bra{\bm s}\,.
\end{align}

In the case of the bipartite qubit system ($N=2$), the RSB Hamiltonian in the basis $\{ \ket{++}, \ket{+-}, \ket{-+}, \ket{--} \}$ is expressed as 
\begin{align}
    \hat H_{0,\Delta, \lambda}
    &=
        \begin{bmatrix}
            \hat H_{++} & 0 & 0 & 0 \\
            0 & \hat H_{+-} & 0 & 0 \\
            0 & 0 & \hat H_{-+} &0 \\
            0 & 0 & 0 & \hat H_{--}
        \end{bmatrix}\,, \notag 
\end{align}
where $H_{p,q}$, $p,q \in \{ \pm \}$ is given by \eqref{eq:bipartite diagonal Hamiltonians}, and the unitary operator generated by this Hamiltonian is thus 
\begin{align}
    \hat U(t)
    =
        \begin{bmatrix}
            e^{-\ii t \hat H_{++}} & 0 & 0 & 0 \\
            0 & e^{-\ii t \hat H_{+-}} & 0 & 0 \\
            0 & 0 & e^{-\ii t \hat H_{-+}} &0 \\
            0 & 0 & 0 & e^{-\ii t \hat H_{--}}
        \end{bmatrix}\,.\notag 
\end{align}

To obtain the reduced density matrix $\rho(t)$, we need to evaluate the vacuum expectation value $\braket{0|e^{ \ii t \hat H_{\bm s} } e^{ -\ii t \hat H_{\bm r} }|0}$. 
To do this, we utilize the following:  
\begin{align}
    e^{ \ii t \hat H_{\bm s} } e^{ -\ii t \hat H_{\bm r} }
    &=
        e^{ \ii t (\bm s - \bm r) \cdot \bm \Delta }
        e^{ \ii [\theta_{\bm s}(t) - \theta_{\bm r}(t) ] } 
        e^{\ii \Im ( \eta_\bk(t) \bm s \cdot \bm{F_k}, \eta_\bk(t) \bm r \cdot \bm{F_k} )_\hil} \notag \\
    &\times 
    W((\bm s - \bm r)\cdot \bm{F_k} \eta_\bk (t) e^{ \ii \omega_\bk t })\,, 
    \label{eq:product of Hamiltonians}
\end{align}
where 
\begin{align}
    \theta_{\bm s}(t)
    &\coloneqq
        \int_{\R^n} \dd^n \bk\,
        \dfrac{ |\bm s\cdot \bm{F_k}|^2 }{\omega_\bk^2} (\sin(\omega_\bk t) - \omega_\bk t)\,,\\
    (f_\bk, g_\bk)_\hil
    &\coloneqq
        \int_{\R^n} \dd^n \bk\,
        f_\bk^* g_\bk \,, \\
    \eta_\bk(t) 
    &\coloneqq
        \dfrac{2 e^{-\ii \omega_\bk t/2} }{ \omega_\bk }
        \sin(\omega_\bk t/2)\,.
\end{align}
This can be derived straightforwardly from Lemma~1 in \cite{Tjoa2024interacting}, which employs a Lie algebraic technique (see Sec.~7.4 in \cite{gilmore2008lie}) to decompose $e^{\ii t \hat H_{\bm s}}$ into products of exponentials: 
\begin{align}
    e^{\ii t \hat H_{\bm s}}
    &=
        e^{ \ii t \bm s\cdot \bm \Delta }
        e^{\ii \theta_{\bm s}(t)}
        e^{\ii t \hat H_{\phi,0}} 
        e^{ \ii \hat \phi( \eta_\bk (t) \bm s \cdot \bm{F_k} ) }\,.
\end{align}
Realizing that $e^{ \ii \hat \phi( \eta_\bk (t) \bm s \cdot \bm{F_k} ) }$ is an element of the Weyl algebra $W(\eta_\bk (t) \bm s \cdot \bm{F_k})\equiv e^{ \ii \hat \phi( \eta_\bk (t) \bm s \cdot \bm{F_k} ) }$, the Weyl relation \cite{Tjoa2024interacting} $W(f_\bk) W(g_\bk)=e^{-\ii \Im(f_\bk, g_\bk)_\hil} W(f_\bk + g_\bk)$ leads to Eq.~\eqref{eq:product of Hamiltonians}. 
From Eq.~\eqref{eq:product of Hamiltonians}, we have the reduced density matrix for the $N$ gapless qubits: 
\begin{align}
    \rho(t)
    &=
        \dfrac{1}{2^N}
        \sum_{\bm s, \bm r \in \{ \pm \}^N} 
        \ketbra{\bm r}{\bm s}
        e^{ \ii t (\bm s - \bm r) \cdot \bm \Delta }
        e^{ \ii [\theta_{\bm s}(t) - \theta_{\bm r}(t) ] } \notag \\
        &\times 
        e^{\ii \Im ( \eta_\bk(t) \bm s \cdot \bm{F_k}, \eta_\bk(t) \bm r \cdot \bm{F_k} )_\hil} 
        \omega_0(
            W(
                (\bm s - \bm r)\cdot \bm{F_k} \eta_\bk^* (t) 
            )
        ), 
    \label{eq:final density matrix general}
\end{align}
where $\omega_0(W(f_\bk))\equiv \braket{0|e^{\ii \hat \phi(f_\bk)}|0}$ with $\omega_0: \mathcal W(\mfd) \to \C$ being the state corresponding to the Fock vacuum $\ket{0}$. 
Practically, this can be evaluated in the following way: 
\begin{align}
    \omega_0(W(f_\bk))
    &=
        \exp 
        \kako{
            -\dfrac{1}{2} 
            \int_{\R^n} \dd^n \bk\,|f_\bk|^2
        }.
\end{align}
Needless to say, a concrete expression of $\rho(t)$ requires an explicit form of $F_\bk^{(j)}$. In the following subsections, we obtain the expression of each element in $\rho(t)$ for the bipartite and tripartite cases when the background spacetime geometry is flat.

\section{Entanglement dynamics between two detectors}
\label{sec: entanglement-two}

In this section, we revisit the scenario of entanglement dynamics between two UDW detectors interacting with a scalar field initialized in the uncorrelated ground state of the noninteracting theory ($\lambda_j=0$) $\ket{\psi(0)} = \ket{g\ts{A}}\ket{g\ts{B}}\ket{0}$. 
It is instructive to survey the relevant known results for two original UDW detectors, all of which were done in the interaction picture:
\begin{enumerate}[leftmargin=*,label=(\roman*)]
    \item Perturbatively, spacelike-separated detectors with an interaction Hamiltonian not commuting with their free Hamiltonian can be entangled starting from an initially separable state. 
    This is the well-known entanglement harvesting protocol \cite{pozas2015harvesting,pozas2016entanglement,reznik2003entanglement,reznik2005violating,Valentini1991nonlocalcorr}, and it has been studied extensively in various scenarios (curved spacetimes, different detector models, etc.). 

    \item Perturbatively, detectors can also be entangled when the interaction is causally connected \cite{pozas2015harvesting}; however, this entanglement can be independent of the vacuum entanglement of the field state, depending on a wave property known as the strong Huygens principle \cite{tjoa2021harvesting, hector2022optimize, Caribe2023harvestingBH}. 

    \item Nonperturbatively, it is known that \textit{delta-coupled} detectors cannot get entangled at spacelike separation if it is ``simple-generated'' \cite{Simidzija2018no-go} (i.e., each detector interacts instantaneously \textit{once} \cite{tjoa2023nonperturbative}), but can get entangled otherwise.  
    
    \item Nonperturbatively, gapless detector models cannot extract entanglement at spacelike separation, based on arguments about entanglement-breaking channels \cite{Simidzija2018no-go,pozas2017degenerate}. 
    Consequently, initially separable detectors remain separable after interaction. 
    
    \item Nonperturbatively, it was shown in \cite{perche2024closedform} for massless fields in (3+1)D Minkowski spacetime that gapless models can get entangled if the interaction time is long enough for the detectors to be in causal contact.\footnote{It was also shown in \cite{perche2024closedform} 
    that if the coupling is also sufficiently weak, the two-detector dynamics can be approximated in some sense by an entangling unitary. The latter is essentially the same spirit as the well-known semiclassical optics calculation where the electromagnetic field is not quantized \cite{scully1997quantum,vedral2005modern}.} 
    This analysis was partially motivated to better understand gravity-induced entanglement setups \cite{emm2023gravityinducedentanglement}.

\end{enumerate}
There are countless results in the literature involving similar setups but different detector trajectories, boundary conditions, spacetime geometries, initial states, etc., that are only tangentially relevant for our purposes, so we do not aim for an exhaustive list. 

Our initial goal is to complement the analysis in the literature concerning (ii), (iv) and (v). That is, we carry out a fully nonperturbative analysis of entanglement dynamics between two gapless detectors that are in causal contact in scenarios where the strong Huygens principle is violated, i.e., accounting for the effects of mass and spacetime dimensions. As we show below, pursuing this goal has a good payoff: it turns out that the entanglement dynamics is actually richer and nontrivial once we include mass and spacetime dimension as tunable parameters.

\begin{figure*}[tp]
\centering
\includegraphics[width=\linewidth]{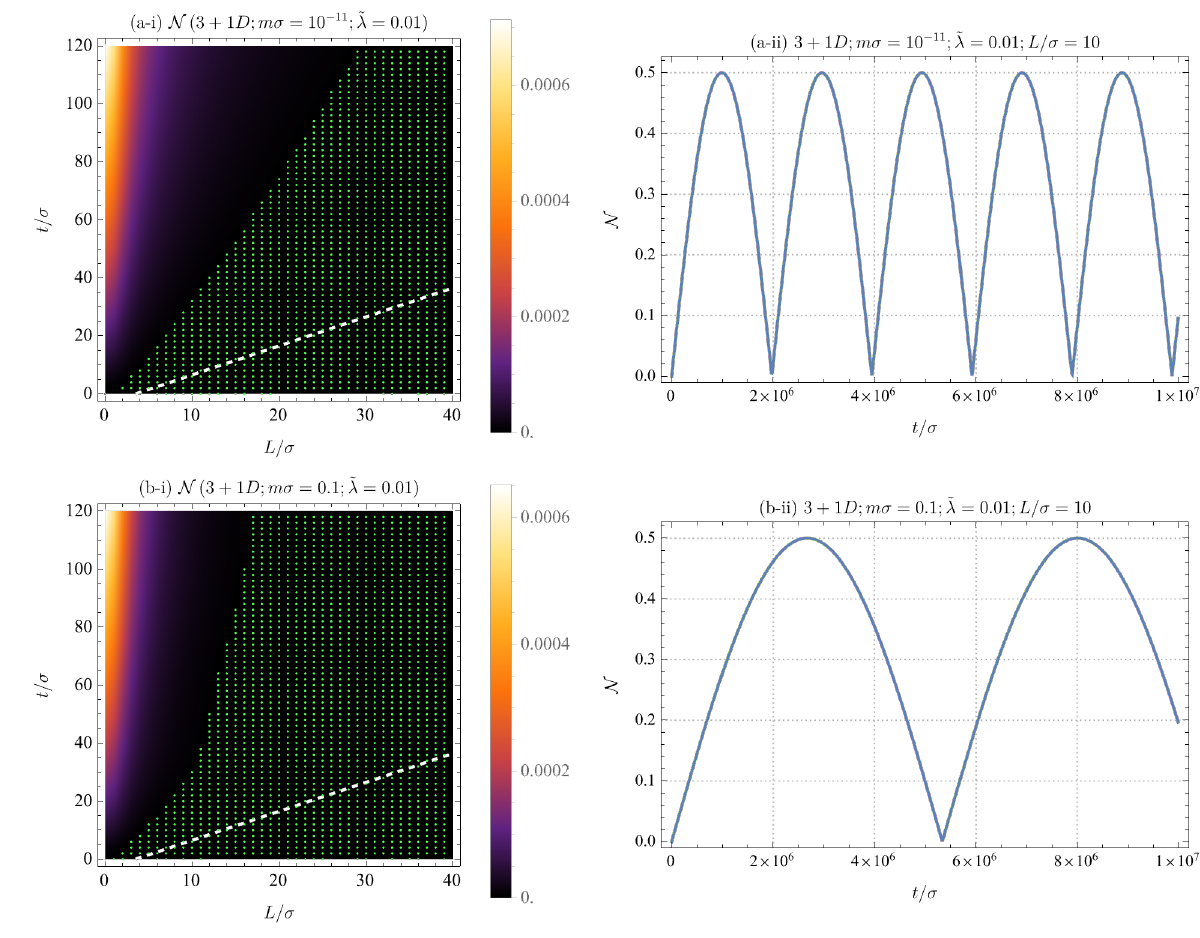}
    \caption{Density plots [(a-i) and (b-i)] of negativity $\mathcal{N}$ between two inertial detectors weakly coupled to the field ($\tilde \lambda =0.01$) with $\Delta\ts{A}=\Delta\ts{B}=0$ in $(3+1)$-dimensional Minkowski spacetime when the mass of the field is (a) $m \sigma=10^{-11}$ and (b) $m \sigma=0.1$. 
    The green dots depict the region where $\mathcal{N}=0$, and the white dashed line represents the light cone from the edge of the detector (which starts from $L/\sigma\approx 3.5$ at $t=0$), whose COM is located at the origin of the diagrams. 
    On the right-hand side [(a-ii) and (b-ii)], the time dependence of $\mathcal{N}$ at $L/\sigma=10$ is depicted. 
}
\label{fig:4Ddensityweak}
\end{figure*}

We first recall that the original interest in the so-called entanglement harvesting protocol is fundamental in nature, since it was to see whether vacuum entanglement can be detected at all \cite{pozas2015harvesting,pozas2016entanglement, Valentini1991nonlocalcorr, reznik2003entanglement, reznik2005violating}. For this reason, it is immaterial that the amount of entanglement is extremely small or whether the calculations are only perturbative in nature, as the main purpose is as a sort of ``witness'' for vacuum entanglement in the field. From a theoretical standpoint, there is nothing very surprising about the protocol itself: what is nontrivial is whether it is detectable in practice. Several experimental proposals have been given, from the early proposal involving a linear ion trap \cite{retzker2005iontrap} to the recent superconducting circuits \cite{teixido2025towards}, with a recent experimental realization via an electro-optics sampling experiment \cite{lindel2024harvestingexp}.

On the other hand, if the goal is to generate entanglement, the quality matters, and it is somewhat immaterial that the entanglement is extracted from the vacuum of the field (which would be very small in general). 
Indeed, physical intuition tells us that it is generically unsurprising for two quantum systems A and B to get entangled in general if we allow them to have direct interactions, and one has to engineer the interactions and the initial states to avoid getting entangled. 
The same is true if we allow induced interactions, i.e., the two systems interact with a common third party C that induces effective interactions between the two. Essentially, C mediates the interaction by propagating signals from A to B (it could be phonons, photons, etc.). 
Consequently, the physically relevant question is how the entanglement is generated and how good it can be as a function of parameters. 
This necessitates a nonperturbative approach, as perturbation theory only allows a perturbatively small amount of entanglement.

With the above perspectives in mind, we would like to better understand how entanglement is generated between two initially uncorrelated detectors interacting with the vacuum state of a noninteracting scalar field as a function of time, separation, mass, and spacetime dimensions. 
The gapless RSB model has two advantages over the standard UDW model in that (i) it can be defined rigorously \cite{Tjoa2024interacting,spohn1989spinboson,amann1991spinboson,fannes1988equilibrium} and (ii) the causal relationship can be defined sharply without any (exponential) tails from the original UDW model employing Gaussian (adiabatic) switching. 
Our results will be presented numerically in what follows. 
All results are plotted with the detector's effective size $\sigma$ as a reference length scale, and we use the dimensionless coupling constant $\tilde{\lambda}\equiv \lambda \sigma^{-(n-3)/2}$. 
The initial state of the total system at $t=0$ is an uncorrelated state $\ket{\psi(0)} = \ket{g\ts{A}}\ket{g\ts{B}}\ket{0}$, where $\ket{g_j}$ is the ground state of $\hat \sigma_z^{(j)}$. 
Fields are considered massless when $m\sigma = 10^{-11}$ (IR-regulated) with numerically verified convergence. 
We note that a nonzero gap $(\Delta_j \neq 0)$ does not alter the dynamics at all, so these results apply equally well to generic $\Delta_j$. 
Hence, in what follows we set $\Delta_j  = 0$.

\begin{figure*}[tp]
\centering
\includegraphics[width=\linewidth]{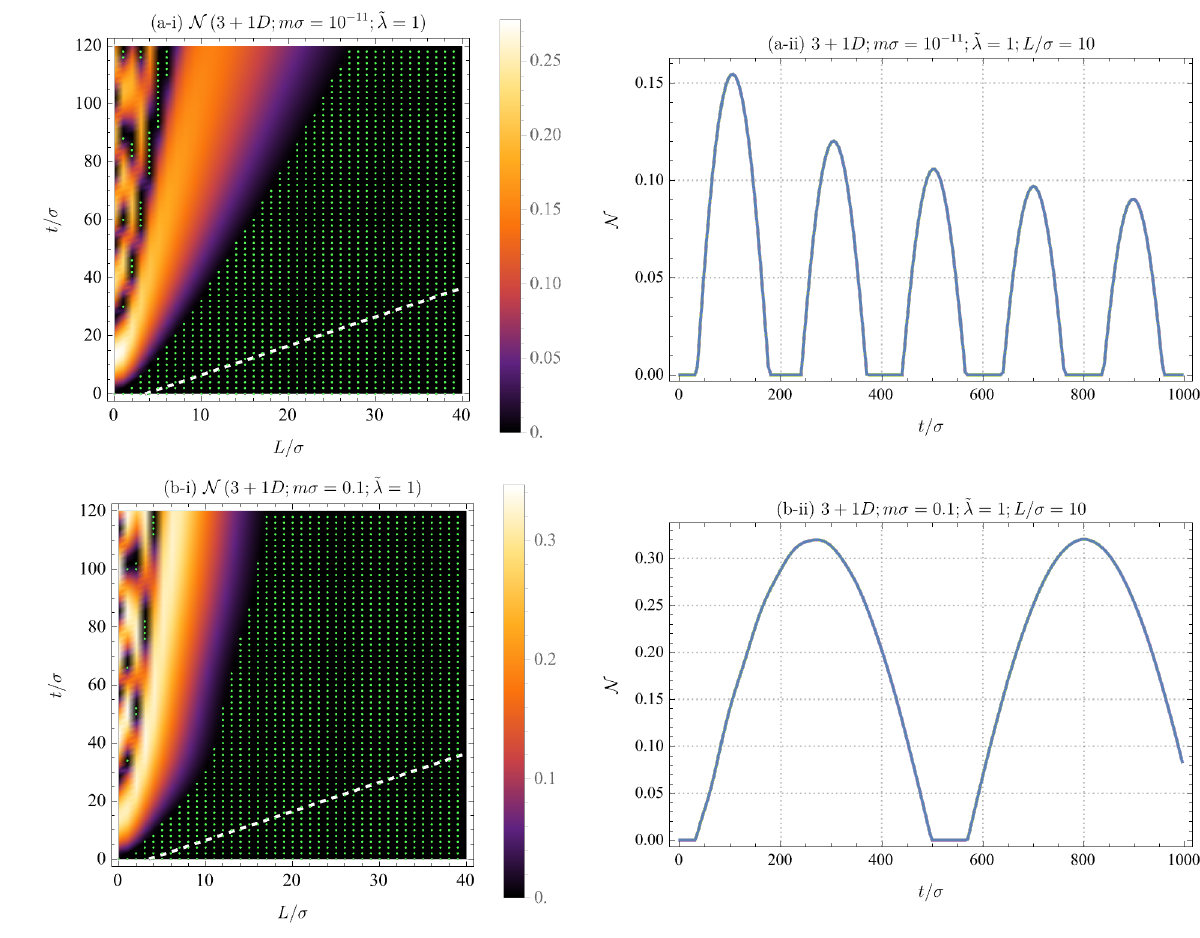}
    \caption{Density plots [(a-i) and (b-i)] of negativity $\mathcal{N}$ between two inertial detectors strongly coupled ($\tilde \lambda =1$) to the field in $(3+1)$-dimensional Minkowski spacetime when the mass of the field is (a) $m \sigma=10^{-11}$ and (b) $m \sigma=0.1$. 
    The green dots depict the region where $\mathcal{N}=0$, and the white dashed line represents the light cone from the edge of the detector (which starts from $L/\sigma\approx 3.5$ at $t=0$), whose COM is located at the origin of the diagrams. 
    On the right-hand side [(a-ii) and (b-ii)], the time dependence of $\mathcal{N}$ at $L/\sigma=10$ is depicted. 
}
\label{fig:4DdensityStrong}
\end{figure*}

In Fig.~\ref{fig:4Ddensityweak}, we plot the negativity of entanglement $\mathcal N$ as a function of separation $L$ between the detectors' centers of mass and the duration of interaction $t$ for the weak coupling regime $\tilde{\lambda}=0.01$, with (a) showing the massless regime and (b) showing the massive regime. 
In both massless and massive cases, we see that the UDW-type coupling can be used to generate bipartite states close to being maximally entangled.\footnote{For the massless case, the fact that we can generate a significant amount of entanglement can already be inferred from \cite{perche2024closedform} with suitable adjustment and rescaling (to account, for instance, the Gaussian switching).} 
However, perhaps surprisingly, both Fig.~\ref{fig:4Ddensityweak}(a-i) and (b-i) show that the light cone (white dashed line) has a much gentler slope than the black boundary which marks the beginning of nonzero entanglement---for convenience, we call these ``entanglement cones.'' 
In other words, a significant amount of entanglement can only be generated \textit{very deep into the light cone}. 
This is to be contrasted with the standard perturbative result for the UDW model, where entanglement peaks in general \textit{at} the light cone \cite{tjoa2021harvesting}. 
The entanglement cone is not linear in the $t$-$x$ plane, and we see that increasing the mass makes the entanglement cone steeper (i.e., it takes a much longer time to reach maximum entanglement).

\begin{figure*}[tp]
\centering
\includegraphics[width=\linewidth]{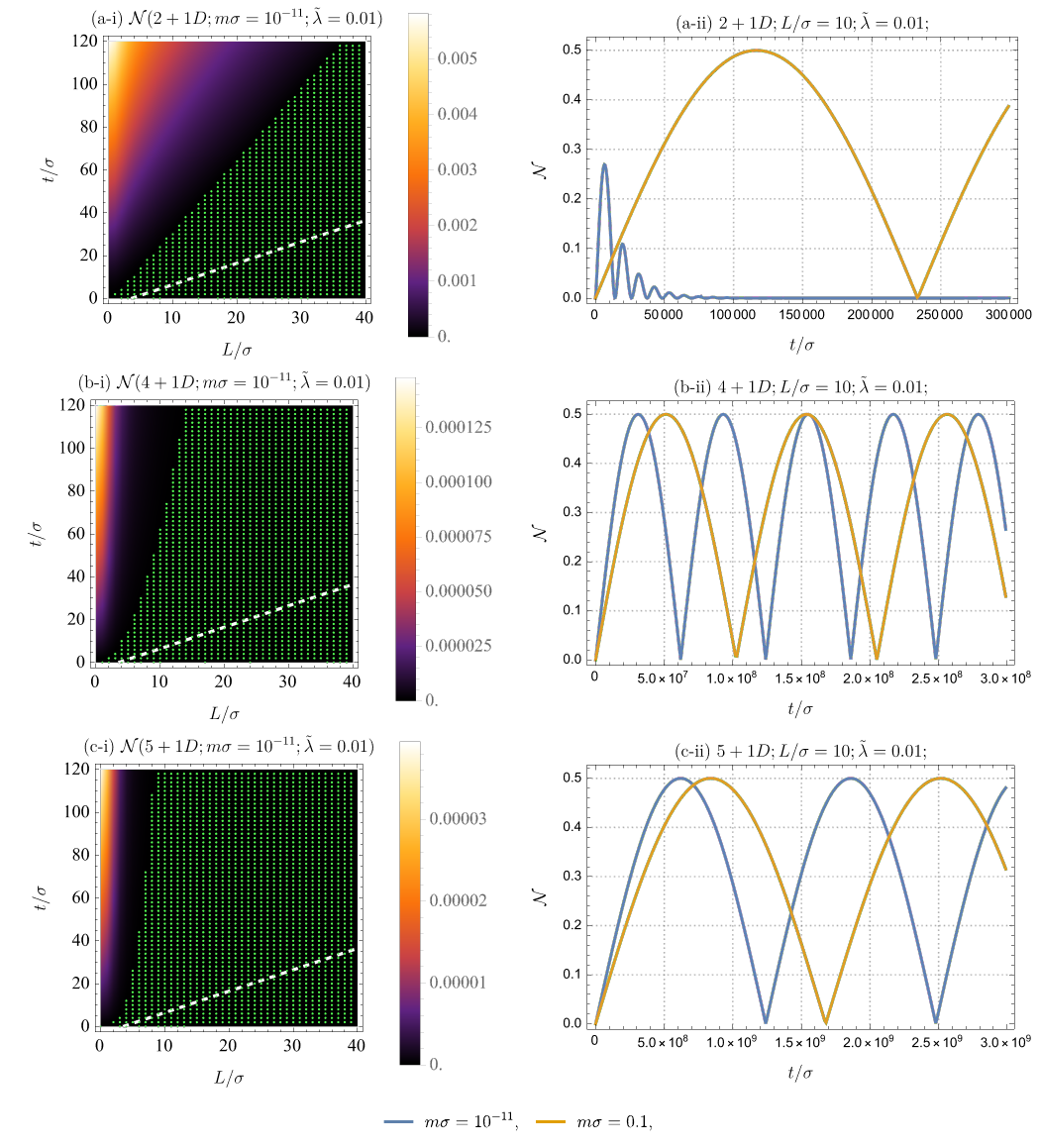}
    \caption{
    Density plots [(a-i), (b-i), and (c-i)] of negativity $\mathcal N$ between two inertial detectors weakly coupled ($\tilde \lambda=0.01$) to the nearly massless field in various spacetime dimensions. 
    Here, $\tilde \lambda\equiv \lambda \sigma^{-(n-3)/2}$ for $(n+1)$-dimensional Minkowski spacetime. 
    On the right panel [(a-ii), (b-ii), and (c-ii)], the time dependence of $\mathcal N$ at $L/\sigma=10$ with $m \sigma=10^{-11}$ (blue plots) and $m\sigma=0.1$ (orange plots) are depicted. 
    }
\label{fig:FigOtherDims}
\end{figure*}

In Fig.~\ref{fig:4DdensityStrong}, we plot the same scenario in (3+1)D for the strong coupling regime $\tilde \lambda =1$. 
In the strong coupling regime, we observe that the entanglement cone is \textit{unchanged} compared to the weak coupling; however, the maximum amount of entanglement is much lower in both massless and massive cases. 
In fact, we observe that not only does the massless entanglement decay in time, but also its maximum is \textit{lower} than that for the massive field. 
In other words, mass \textit{enhances} the amount of entanglement generated for the two detectors. 
Note that because the RSB model used here is gapless, this cannot be explained by any resonance-type argument (which works in the perturbative regime as in \cite{hector2022optimize}). 
Importantly, this does not contradict any intuition about ``exponential decay'' associated with mass: mass contributes to the exponential decay of correlations \textit{as a function of separation $L$}, but not as a function of temporal duration $t$---this can be checked, for instance, by looking at the asymptotic behavior of the Wightman two-point functions in the large-$t$ or large-$L$ limit.

\begin{figure*}[tp]
\centering
\includegraphics[width=\linewidth]{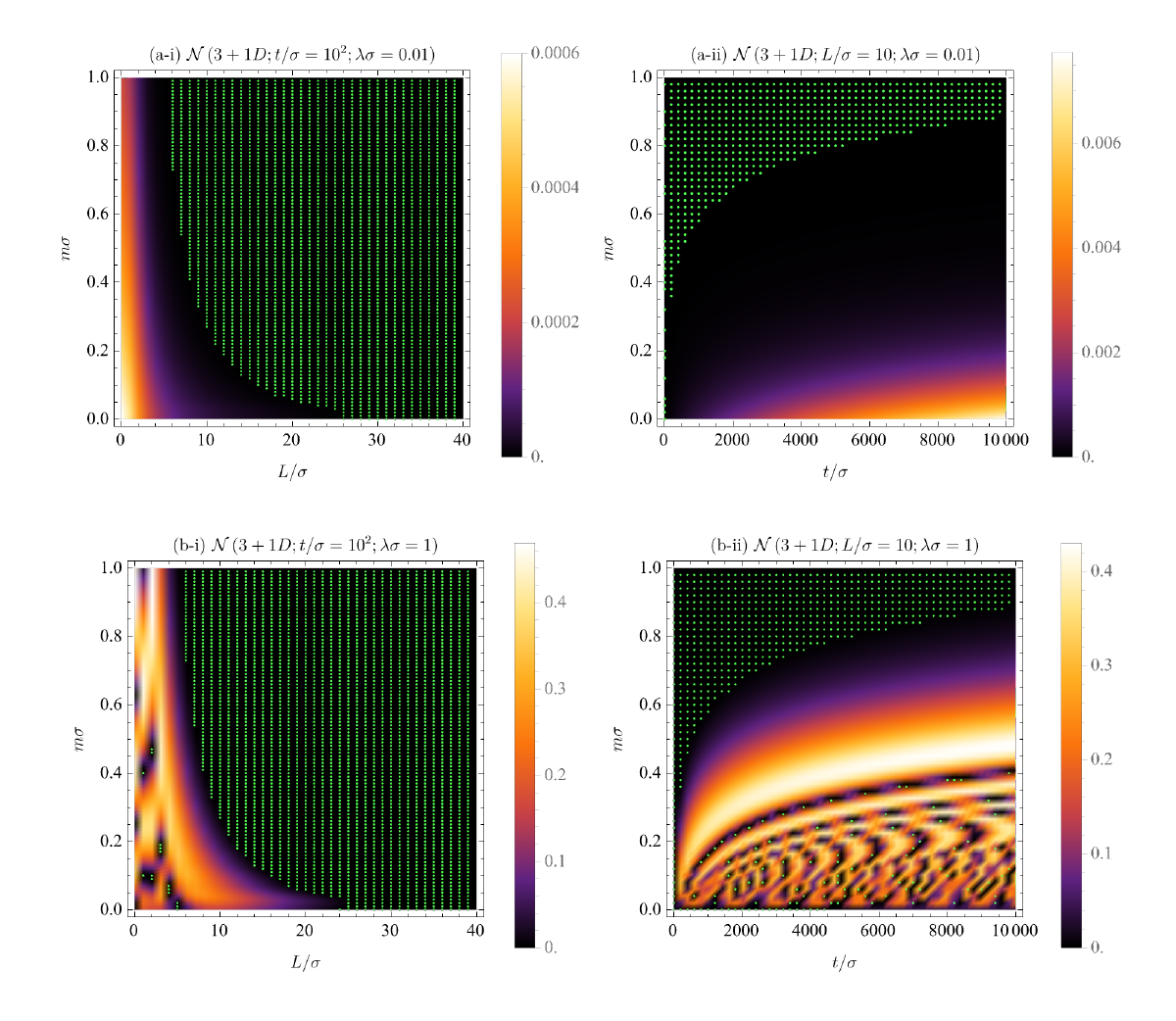}
    \caption{
    Density plots of negativity $\mathcal N$ between two inertial detectors in $(3+1)$-dimensional Minkowski spacetime, showing how mass $m$ of the field contributes to $\mathcal N$. 
    (a) and (b) correspond to the weak coupling $\tilde \lambda=0.01$ and strong coupling $\tilde \lambda =1$ cases. 
    (a-i) and (b-i) depict the dependence on $m\sigma$ and $L/\sigma$ at a fixed time $t/\sigma=10^2$, whereas (a-ii) and (b-ii) show the dependence on $m\sigma$ and $t/\sigma$ at a fixed distance $L/\sigma=10$. 
    }
\label{fig:4DdensityVsMass}
\end{figure*}

At this stage, it is unclear if this phenomenon is tied to the strong Huygens principle. 
If it were, then the same result should be expected by varying the spacetime dimensions rather than mass. 
We show that this is intrinsically due to mass rather than the violation of the strong Huygens principle in Fig.~\ref{fig:FigOtherDims} for the weak coupling regime. 
Here we can see more clearly in the rightmost Fig.~\ref{fig:FigOtherDims}(a-ii)--(c-ii) that mass generically improves entanglement generation, in that it always produces at least as much entanglement as the massless case, as is the case in (4+1)D and (5+1)D. 
In (2+1)D, the separation between the massless and massive cases is huge: we see that the massless field can generate only half the negativity of the massive case. 
What is true in general, however, is that massive fields attain maximum entanglement \textit{much more slowly} than the massless field. 
We note in passing that the entanglement cone is steeper as the spacetime dimension is increased---this is itself not surprising, as this can be explained by the fact that the Wightman two-point functions have stronger decay (in both space and time directions) in higher dimensions. 
It implies that generating entanglement for a given separation is harder in higher spacetime dimensions.

Now that we have singled out mass as the factor that contributes to the enhancement of entanglement between the two detectors, it is useful to see how the entanglement pattern changes with mass when we fix either the spatial separation or the duration of interaction, as shown in Fig.~\ref{fig:4DdensityVsMass}. 
Here we see that the properties of spatial and temporal correlations of the field are indeed asymmetric and influence the negativity differently.\footnote{Note that the asymmetry in the way correlations decay in the spatial and temporal directions is not unique to relativistic scalar fields: it is a well-known phenomenon in quantum many-body and condensed matter physics, which underlies why quantum many-body dynamics is a hard problem even for systems with exponentially decaying {spatial} correlations.} 
Notice that in Figs.~(a-i) and (b-i), we understand in what sense massless fields are more helpful: the larger the mass, the faster the negativity vanishes as a function of $L$, hence massless fields have ``larger reach'' in spatial separation as negativity vanishes at larger separation. 
Figures \ref{fig:4DdensityVsMass}(a-ii) and \ref{fig:4DdensityVsMass}(b-ii) complement our earlier observations: mass generically improves the maximum amount of entanglement generated on the detectors, in this case for the strong-coupling regime relative to the weak-coupling one.

We summarize the main results as follows:
\begin{enumerate}[label=(\roman*),leftmargin=*]
    \item Nonperturbatively, gapless models can indeed generate entanglement between the two detectors, but they generically do so \textit{very deep into the light cone}. 
    This is in contrast to the standard perturbative result where entanglement is generated starting at the light cone (in other words, light cone and entanglement cone coincide) \cite{tjoa2021harvesting}. 
    This is a stronger result than the no-go theorem in \cite{Simidzija2018no-go}, which only proves that entanglement cannot be generated at spacelike-separation, since we showed that their interactions being in causal contact is not enough. 
    
    \item Mass can improve entanglement generated relative to the massless case in terms of the amount of negativity, sometimes approaching maximal entanglement $\mathcal{N}\approx 0.5$ when the massless field cannot. 
    This is independent of the violation of the strong Huygens property. 
    However, in general, massive fields only achieve maximum negativity at a much longer interaction time, consistent with the fact that the entanglement cone is steeper with increasing mass. 

    \item There is a clear asymmetry between the decay of correlations in the spatial and temporal directions, which manifests in the difference in how the negativity of the two detectors falls off as a function of time and spatial separation for different mass parameters. 
\end{enumerate}

\section{Entanglement dynamics between three detectors}
\label{sec: entanglement-three}

In this section, we study the entanglement dynamics between three gapless RSB-type UDW detectors interacting with a scalar field initialized in the uncorrelated ground state  of the noninteracting theory ($\lambda_j=0$) given by 
\begin{align}
    \ket{\psi(0)} = \ket{g_{\ts{A}}}\ket{g_{\ts{B}}}\ket{g_{\ts{C}}}\ket{0}\,.
\end{align}
For simplicity, the three detectors are assumed to be static and placed symmetrically at the vertices of an equilateral triangle, each separated by a distance $L$ from one another \cite{mendez2022tripartite}. 
This simplification is also helpful when trying to isolate the physics associated with the detector-field dynamics from the nonuniformity introduced by the asymmetric placement of the detectors.

Due to the relatively painful amount of computations, tripartite dynamics in the standard UDW model is already sparse. 
The relevant known results so far are
\begin{enumerate}[label=(\roman*),leftmargin=*]
    \item Perturbatively, it was shown in Minkowski spacetime that three UDW detectors can extract tripartite entanglement at spacelike separation in various spatial configurations, not only the equilateral one \cite{mendez2022tripartite}. 
    This is also demonstrated for (2+1)D BTZ black hole spacetime \cite{membrere2023tripartite}. 
    The authors used $\pi$-tangle \cite{ou2007monogamytripartite} that serves as a lower bound for the amount of genuine tripartite entanglement in mixed states of three qubits. 
    It is noteworthy that $\pi$-tangle is not an entanglement measure on its own. There were some earlier studies that show W-type entanglement in entanglement harvesting-type setups (i.e., for spacelike-separated detectors) \cite{silman2005distil}.

    \item Nonperturbatively, it was shown for Minkowski  \cite{mendez2023tripartite}  that three UDW detectors can get entangled for \textit{delta-coupled} detector model, also using $\pi$-tangle as a witness of genuine tripartite entanglement. 
    They show that there exist parameter regimes for which the tripartite entanglement is GHZ-type.
\end{enumerate}
As before, we do not aim for an exhaustive list. 
From these, we believe that the long-time regime involving the gapless model has not been fully understood prior to this work. 
Let us remark that the regularity requirements for the three-detector RSB model follow directly from the general case in Sec.~\ref{sec: regularity}.

It is well known that constructing entanglement measures for multipartite entanglement is highly nontrivial. 
Even in the case of three qubits, it is known that it is not possible to have a single entanglement measure that unambiguously tells whether one state is more entangled than another \cite{acin2000schmidt}. 
This is related to the fact that for multipartite entanglement one can have inequivalent but \textit{genuine multipartite entanglement} (GME): the simplest examples being the GHZ and W states, which are in some sense ``maximally entangled'' (see, e.g., \cite{eberly2021triangle}). 
For our purposes, this issue is complicated by the fact that if we consider the RSB model with three detectors, then the reduced density matrix of the detectors is mixed: for mixed states, constructing a faithful entanglement measure is generically difficult, and none of the convex-roof extensions (if at all) are easy to compute for mixed states.

In order to analyze the tripartite entanglement, we follow \cite{mendez2022tripartite, mendez2023tripartite, membrere2023tripartite} and consider the \textit{$\pi$-tangle}, defined as 
\begin{subequations}
\begin{align}
    \pi
    &\coloneqq
        \dfrac{\pi\ts{A} + \pi\ts{B} + \pi\ts{C}}{3}\,, \\
    \pi\ts{A}
    &\coloneqq
        \mathcal{N}\ts{A(BC)}^2 - \mathcal{N}\ts{A(B)}^2 - \mathcal{N}\ts{A(C)}^2\,, \\
    \pi\ts{B}
    &\coloneqq
        \mathcal{N}\ts{B(CA)}^2 
        - \mathcal{N}\ts{B(C)}^2 
        - \mathcal{N}\ts{B(A)}^2\,, \\
    \pi\ts{C}
    &\coloneqq
        \mathcal{N}\ts{C(AB)}^2 
        - \mathcal{N}\ts{C(A)}^2 
        - \mathcal{N}\ts{C(B)}^2\,.
\end{align}
\end{subequations}
Here, 
\begin{subequations}
\begin{align}
    \mathcal{N}\ts{A(BC)}
    &\coloneqq
        \frac{|| \rho\ts{ABC}^{\intercal\ts{A}} ||_1 -1}{2} \,, \\
    \mathcal{N}\ts{A(B)}
    &\coloneqq
        \frac{|| \rho\ts{AB}^{\intercal\ts{A}} ||_1 -1}{2} \,, \\
    \mathcal{N}\ts{A(C)}
    &\coloneqq
        \frac{|| \rho\ts{AC}^{\intercal\ts{A}} ||_1 -1}{2} \,,
\end{align}
\end{subequations}
are the negativities of the three-detector system, where $||\cdot||_1$ is the trace norm. The $\pi$-tangle quantifies the tripartite entanglement if the joint state of the detectors $\rho\ts{ABC}$ is pure due to the Coffman-Kundu-Wootters inequality for negativity, $\pi_i \geq 0$, $\forall i\in \{ \AAA, \BB, \CC \}$ \cite{CKW.distributed.2000, ou2007monogamytripartite}. Unfortunately, this property fails when $\rho\ts{ABC}$ is mixed, and the $\pi$-tangle serves only as a lower bound for the tripartite entanglement. This suggests that the $\pi$-tangle for mixed states could be zero or negative, even if there exists tripartite entanglement among the detectors. Nevertheless, tripartite entanglement is guaranteed at least when the $\pi$-tangle is positive. For this reason, in what follows we take $\max\{ \pi, 0 \}$ and examine the parameter space with positive $\pi$-tangle.

This time, it will be useful to first summarize our results for the tripartite entanglement dynamics involving the gapless RSB model:
\begin{enumerate}[label=(\roman*),leftmargin=*]
    \item There exists a regime in the parameter space for which the three detectors have \textit{genuine tripartite entanglement} of the ``GHZ type,'' in the sense that any bipartition has zero entanglement. 
    However, this regime is generically very narrow and has a very small $\pi$-tangle, thus it is difficult to certify that the tripartite entanglement is significant.

    \item A large portion of the parameter space allows for tripartite entanglement such that any bipartition has nonzero entanglement. 
    Similar to the bipartite case, the mass of the field tends to amplify the $\pi$-tangle, but there is no natural way to discern the nature of the tripartite entanglement.
\end{enumerate}
 We will provide arguments for why this is the case and put the existing known results in perspective.

\begin{figure*}[t]
\centering
\includegraphics[width=\linewidth]{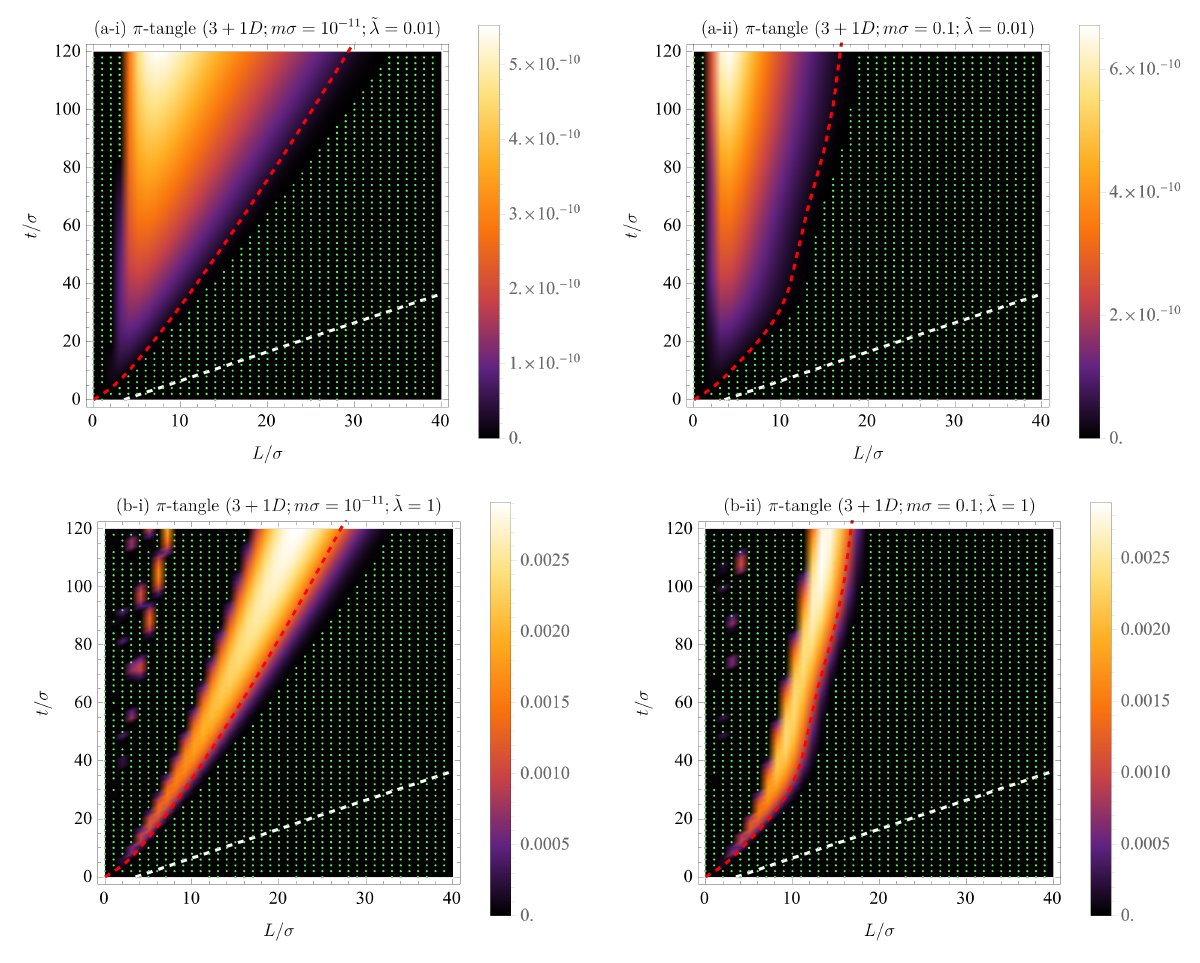}
    \caption{
    Density plots of the positive $\pi$-tangle among three inertial detectors in the equilateral triangular configuration in $(3+1)$-dimensional Minkowski spacetime when the detectors are (a) weakly coupled, and (b) strongly coupled to the field. 
    (a-i) and (b-i) show the nearly massless ($m \sigma =10^{-11}$) scenario, whereas (a-ii) and (b-ii) correspond to the massive case with $m \sigma=0.1$. 
    The green dots depict the region where the $\pi$-tangle vanishes. 
    The white dashed line represents the light cone from the edge of the detector as before, while the red dashed curve corresponds to the boundary of the bipartite entanglement cone depicted in Figs.~\ref{fig:4Ddensityweak} and \ref{fig:4DdensityStrong}. 
    }
\label{fig:pitangle3DVer2}
\end{figure*}

Our argument is based on a numerical analysis in Fig.~\ref{fig:pitangle3DVer2} for both weak and strong coupling as well as massless and massive fields. 
The white dashed line represents the light cone. 
What we show here is that on the one hand, the black region below the red line in Fig.~\ref{fig:pitangle3DVer2} corresponds to the case for which the three detectors have \textit{genuine tripartite entanglement} of the ``GHZ type,'' in the sense that any bipartition has zero entanglement. 
However, this region is generically very narrow, thus one has to fine-tune the arrangement of the detectors to obtain this. 
Furthermore, since the  $\pi$-tangle is very small, it is difficult to certify that this genuine tripartite entanglement is actually significant. 
On the other hand, the area bounded by the green dots above the red line shows the regime that allows for tripartite entanglement with nonzero bipartite entanglement. 
We also observe that the mass of the field tends to amplify the $\pi$-tangle; however, there is no natural way to obtain precise information about the nature of the tripartite entanglement. 
Indeed, we checked that none of the usual tripartite entanglement witnesses for mixed states works \cite{Sanpera2001tripartite}, as we verified them numerically and all of them are inconclusive (i.e., they say nothing about the tripartite entanglement content). 
These results seem to suggest that there are significant limitations in using gapless UDW or RSB-type models to learn about both intrinsic multipartite entanglement and multipartite entanglement generation in QFT.

A natural question to ask is how this result is compatible with the observations in \cite{mendez2023tripartite}, which essentially claim that there is genuine tripartite entanglement of GHZ type in Minkowski spacetime and BTZ spacetime for the vacuum state of a massless scalar field with a delta-coupled model. 
Strictly speaking, while there are some similarities between the delta-coupled model and the gapless model \cite{tjoa2023nonperturbative}, the results in \cite{mendez2023tripartite} are not ideal as they showed genuine tripartite entanglement only for very small detector separation, effectively with significant overlap of their Gaussian smearing functions. 
In this regime, one can question the validity of the standard UDW model. 
This is because very closely spaced atoms may experience interatomic interactions that the UDW model we considered here does not capture. 
In this sense, our results are more reliable as they do not rely on significant overlap of the spatial smearing and still display genuine tripartite entanglement of GHZ type.

\section{Discussion and outlook}

In this work, we have studied nonperturbatively the entanglement generation between two emitters in an exactly solvable relativistic variant of the spin-boson model, equivalent to the time-independent formulation of the (gapless) Unruh-DeWitt model. 
This formulation allows for a rigorous basis of the regularity of the model in the UV and in the IR and allows a slightly cleaner analysis of the causal behavior of the interactions.

We showed that (i) (highly) entangled states of the two emitters (detectors) require interactions \textit{very deep into the light cone}, strengthening the no-go result of \cite{Simidzija2018no-go} that says there can be no entanglement for spacelike-separated interactions; (ii) the mass of the field generically improves the maximum amount of entanglement generated on the detectors at the expense of longer interaction times, and furthermore, we argued that this is not fundamentally due to the strong Huygens principle; (iii) while it is possible to find regimes with genuine GHZ-like tripartite entanglement, this regime turns out to be really small. 
Furthermore,  it is difficult to find regimes where tripartite entanglement can be easily shown to be significant or can be easily classified into different inequivalent types, refining the delta-coupling result in \cite{mendez2023tripartite}. 
Result (iii), in particular, suggests that probing the multipartite entanglement of a relativistic quantum field nonperturbatively requires either different probe-based techniques or variants of the UDW model.

In a way, our results suggest that despite the nonperturbative and simple nature of the model, the gapless variant of the UDW model is not good enough once we are interested in {obtaining more precise information about multipartite entanglement, whether to study the intrinsic multipartite entanglement in the field or the ability for relativistic quantum fields to mediate entanglement generation}. 
From the quantum-optical standpoint, it means that this model is not very useful for generating highly entangled multipartite states \textit{even if we allow causal contact}, and since this model is unable to extract entanglement from the vacuum state at spacelike separation, we can think of the gapless model as being suitable, for instance, for studying classical communication settings (see, e.g., \cite{landulfo2024broadcast}). 
That said, our work closes the (remaining) gap for the two-detector scenario: we showed that causal contact enables entanglement generation in a maximal sense, but this requires the interaction time to be much larger than the spatial separation of the detectors, in contrast to the known perturbative results \cite{tjoa2021harvesting}.

Three natural questions arise from our work. 
First, is it possible to consider a variant of the setup or the model in such a way that we can probe nonperturbatively {multipartite entanglement} of the vacuum state of the field in a useful way? 
Presumably, there are operator-algebraic approaches to studying multipartite entanglement \cite{vanLuijk2024multipartite} of the QFT directly, and this should ideally be complementary. 
Second, can multiple delta-coupled detectors be used more effectively to study entanglement in QFT? The first use case was in \cite{Simidzija2018no-go}, and recently multidelta coupling has been used in single-detector settings \cite{uenaga2025exact,pologomez2024delta}. 
These quickly become unwieldy for a large number of couplings and detectors, but in principle computable. 
{Last but not least, one way to resolve the ambiguity regarding the nature and ``strength'' of the tripartite entanglement is by operational means, i.e., by identifying protocols (e.g., teleportation in or steering in the bipartite case) that require certain kinds and amounts of tripartite entanglement to succeed and see if the tripartite state produced via the RSB interaction works for these protocols}. 
We leave these questions for the future.

\section*{Acknowledgments}

E.T. acknowledges financial support from the Alexander von Humboldt Foundation. K.G.Y. is partially supported by Grant-in-Aid for Research Activity Start-up (Grant No. JP24K22862) and by Grant-in-Aid for JSPS Fellows (Grant No. JP25KJ0048).

\section*{Data availability}
The data that support the findings of this paper are not
publicly available. 
The data are available from the authors upon reasonable request.

\appendix

\section{Joint density matrix for gapless detectors}

\subsection{Two gapless qubits in Minkowski spacetime}
As a concrete example, we restrict ourselves to $(n+1)$-dimensional Minkowski spacetime and consider two static detectors with zero relative velocity. 
For the sake of convenience, we choose the following Gaussian spatial smearing profile: 
\begin{align}
    F^{(j)}(\bx)
    &=
        \dfrac{1}{ (\sqrt{\pi} \sigma)^n }
        e^{ -(\bx - \bx_{j0})^2/\sigma^2}\,,
    \quad
    j\in \{ \AAA, \BB \}\,.
\end{align}
Here, $\sigma>0$ is the characteristic Gaussian width and $\bx_{j0}$ is the center of the smearing function. 
This smearing function is Fourier transformed into 
\begin{align}
    F_\bk^{(j)}
    &\coloneqq
        \int_{\Sigma_0} \dd^n x
        \sqrt{h_0} \lambda_j F^{(j)}(\bx) 
        u_\bk^*(0,\bx) \notag \\
    &=
        \int_{\R^n} \dfrac{ \dd^n x }{ \sqrt{ 2\omega_\bk (2\pi)^n } }
        \lambda_j F^{(j)}(\bx) e^{ \ii \bk \cdot \bx } \notag \\
    &=
        \dfrac{ \lambda_j e^{ -\kk^2 \sigma^2/4 } e^{\ii \bk \cdot \bx_{j,0}} }{ \sqrt{ 2\omega_\bk (2\pi)^n } }\,.
\end{align}

Define the distance $L$ between the detectors as $L\coloneqq |\bx\ts{A,0} - \bx\ts{B,0}|$. 
From Eq.~\eqref{eq:final density matrix general}, the reduced density matrix $\rho\ts{AB}(t)$ in the basis $\{ \ket{++}, \ket{+-}, \ket{-+}, \ket{--} \}$ reads 
\begin{align}
    \rho\ts{AB}(t)
    &=
        \dfrac{1}{4}
        \begin{bmatrix}
            1 & \rho_{12} & \rho_{13} & \rho_{14} \\
            \rho_{12}^* & 1 & \rho_{23} & \rho_{24} \\
            \rho_{13}^* & \rho_{23}^* & 1 & \rho_{34} \\
            \rho_{14}^* & \rho_{24}^* & \rho_{34}^* & 1
        \end{bmatrix}\,,
\end{align}
with 
\begin{subequations}
    \begin{align}
    \rho_{12}
    &=
        e^{ -2 \ii t \Delta\ts{B}  }
        e^{ -\ii \vartheta\ts{AB} }
        \omega_0 ( W( \Fk{B} \eta_\bk^*(t) ) )\,, \\
    \rho_{13}
    &=
        e^{ -2\ii t \Delta\ts{A} }
        e^{ -\ii \vartheta\ts{AB} }
        \omega_0 ( W( \Fk{A} \eta_\bk^*(t) ) )\,, \\
    \rho_{14}
    &=
        e^{ -2\ii t (\Delta\ts{A} + \Delta\ts{B}) }
        e^{ - \Xi\ts{AB} } \notag \\
        &\quad \times 
        \omega_0 ( W( \Fk{A} \eta_\bk^*(t) ) )
        \omega_0 ( W( \Fk{B} \eta_\bk^*(t) ) )
         \,, \\
    \rho_{23}
    &=
        e^{ -2\ii t (\Delta\ts{A} - \Delta\ts{B}) }
        e^{ \Xi\ts{AB} }\notag \\
        &\quad \times 
        \omega_0 ( W( \Fk{A} \eta_\bk^*(t) ) )
        \omega_0 ( W( \Fk{B} \eta_\bk^*(t) ) )
        \,, \\
    \rho_{24}
    &=
        e^{ -2\ii t \Delta\ts{A} }
        e^{ \ii \vartheta\ts{AB} }
        \omega_0 ( W( \Fk{A} \eta_\bk^*(t) ) ) \,, \\
    \rho_{34}
    &=
        e^{ -2 \ii t \Delta\ts{B} }
        e^{ \ii \vartheta\ts{AB} }
        \omega_0 ( W( \Fk{B} \eta_\bk^*(t) ) )\,,
\end{align}
\end{subequations}
where 
\begin{widetext}
\begin{subequations}
\begin{align}
     \omega_0(W(F_\bk^{(j)} \eta_\bk^*(t)))
    &=
        \exp
        \kako{
            -\dfrac{8\lambda_j^2}{(2\pi)^n}
            \dfrac{\sqrt{ \pi^n }}{ \Gamma(n/2) }
            \int_0^\infty \dd \kk 
            \dfrac{ \kk^{n-1} \sin^2(\omega_\bk t/2) }{ \omega_\bk^3 } e^{-\kk^2 \sigma^2/2}
        }\,,\\
        \vartheta\ts{AB}
    &=
        \dfrac{4 \lambda\ts{A} \lambda\ts{B}}{ (2\pi)^n } \sqrt{\pi^n}
        \int_0^\infty \dd \kk\,
        \dfrac{ \kk^{n-1} e^{-\kk^2 \sigma^2/2} }{ \omega_\bk^3 }
        (\sin(\omega_\bk t) - \omega_\bk t)
        _{0}\mathcal F_1
        \kako{
            \dfrac{n}{2};
            -\dfrac{\kk^2 L^2}{4}
        }
        \,, \\
    \Xi\ts{AB}
    &=
        \dfrac{16\lambda\ts{A} \lambda\ts{B}}{(2\pi)^n} \sqrt{\pi^n}
        \int_0^\infty \dd \kk\,
        \dfrac{ \kk^{n-1} e^{ -\kk^2 \sigma^2/2 } }{\omega_\bk^3}
        \sin^2 (\omega_\bk t/2)
        {}_0 \mathcal F_1 
        \kako{
            \dfrac{n}{2};
            -\dfrac{\kk^2 L^2}{4}
        }
        \,.
\end{align} \label{eq:explicit forms for phases}
\end{subequations}
\end{widetext}
Here, ${}_0 \mathcal F_1$ is the regularized generalized hypergeometric function and $\Gamma(z)$ is the Gamma function \cite{NIST:DLMF,abramowitz1965handbook}. We note that the phase term
\begin{align*}
    e^{\ii \Im ( \eta_\bk(t) \bm s \cdot \bm{F_k}, \eta_\bk(t) \bm r \cdot \bm{F_k} )_\hil}=1
\end{align*}
as $\Im ( \eta_\bk(t) \bm s \cdot \bm{F_k}, \eta_\bk(t) \bm r \cdot \bm{F_k} )_\hil=0$ for a real smearing function $F^{(j)}$.\footnote{This is consistent with the fact that the causal propagator $E(f,g)$ is zero when the supports of $f$ and $g$ are causally disconnected, as $\Im(f_\bk, g_\bk)_\hil$ is the ``$\bk$-space version'' of $E(f,g)$. See \cite{Tjoa2024interacting}.}

\subsection{Three gapless qubits in the equilateral triangular configuration}
We now consider three static gapless detectors A, B, and C, in $(n+1)$-dimensional Minkowski spacetime. 
In this tripartite case, we have the freedom to choose the spatial configuration of the detectors. 
The simplest one is the equilateral triangle configuration \cite{mendez2022tripartite, mendez2023tripartite, membrere2023tripartite}, namely, each detector is located at a corner of an equilateral triangle at each time $t$. 
Such a symmetrical configuration simplifies the calculation drastically, as the proper distance between any pair of the detectors is the same $L$.

From Eq.~\eqref{eq:final density matrix general} with $N=3$, the reduced density matrix is an $8\times 8$ matrix, 
\begin{align}
    \rho(t)
    =
    \dfrac{1}{8}
    \begin{bmatrix}
        1 & \rho_{12} & \rho_{13} & \rho_{14} & \rho_{15} & \rho_{16} & \rho_{17} & \rho_{18} \\
        \rho_{12}^* & 1 & \rho_{23} & \rho_{24} & \rho_{25} & \rho_{26} & \rho_{27} & \rho_{28} \\
        \rho_{13}^* & \rho_{23}^* & 1 & \rho_{34} & \rho_{35} & \rho_{36} & \rho_{37} & \rho_{38} \\
        \rho_{14}^* & \rho_{24}^* & \rho_{34}^* & 1 & \rho_{45} & \rho_{46} & \rho_{47} & \rho_{48} \\
        \rho_{15}^* & \rho_{25}^* & \rho_{35}^* & \rho_{45}^* & 1 & \rho_{56} & \rho_{57} & \rho_{58} \\
        \rho_{16}^* & \rho_{26}^* & \rho_{36}^* & \rho_{46}^* & \rho_{56}^* & 1 & \rho_{67} & \rho_{68} \\
        \rho_{17}^* & \rho_{27}^* & \rho_{37}^* & \rho_{47}^* & \rho_{57}^* & \rho_{67}^* & 1 & \rho_{78} \\
        \rho_{18}^* & \rho_{28}^* & \rho_{38}^* & \rho_{48}^* & \rho_{58}^* & \rho_{68}^* & \rho_{78}^* & 1 \\
    \end{bmatrix},
\end{align}
in the basis $\{ \ket{+++}, \ket{++-}, \ket{+-+}, \ket{-++}$, $\ket{+--}$, $\ket{-+-}$, $\ket{--+}, \ket{---}\}$. 
Each element reads 
\begin{align*}
    \rho_{12} 
    &=
        e^{-2\ii t \Delta\ts{C}}
        e^{ \ii (\vartheta\ts{CA} + \vartheta\ts{BC})}
        \omega_0(W( \Fk{C} \eta_\bk^*(t) ))\,, \\
    \rho_{13}
    &=
        e^{-2\ii t \Delta\ts{B}}
        e^{ \ii (\vartheta\ts{AB} + \vartheta\ts{BC}) }
        \omega_0( W( \Fk{B} \eta_\bk^*(t) ) )\,, \\
    \rho_{14}
    &=
        e^{-2\ii t \Delta\ts{A}}
        e^{ \ii (\vartheta\ts{AB} + \vartheta\ts{CA}) }
        \omega_0( W( \Fk{A} \eta_\bk^*(t) ) )
        \,, \\
    \rho_{15}
    &=
        e^{-2\ii t (\Delta\ts{B} + \Delta\ts{C} )}
        e^{ \ii (\vartheta\ts{AB} + \vartheta\ts{CA}) }
        e^{-\Xi\ts{BC}} \notag \\
        &\times
        \omega_0( W( \Fk{B} \eta_\bk^*(t) ) )
        \omega_0( W( \Fk{C} \eta_\bk^*(t) ) )
        \,, \\
    \rho_{16}
    &=
        e^{-2\ii t (\Delta\ts{A} + \Delta\ts{C} )}
        e^{ \ii (\vartheta\ts{AB} + \vartheta\ts{BC}) }
        e^{-\Xi\ts{CA}} \notag \\
        &\times
        \omega_0( W( \Fk{A} \eta_\bk^*(t) ) )
        \omega_0( W( \Fk{C} \eta_\bk^*(t) ) )
        \,, \\
    \rho_{17}
    &=
        e^{-2\ii t (\Delta\ts{A} + \Delta\ts{B} )}
        e^{ \ii (\vartheta\ts{CA} + \vartheta\ts{BC}) }
        e^{-\Xi\ts{AB}} \notag \\
        &\times
        \omega_0( W( \Fk{A} \eta_\bk^*(t) ) )
        \omega_0( W( \Fk{B} \eta_\bk^*(t) ) )
        \,, \\
    \rho_{18}
    &=
        e^{-2\ii t (\Delta\ts{A} + \Delta\ts{B} + \Delta\ts{C} )}
        e^{-\Xi\ts{AB}}
        e^{-\Xi\ts{BC}}
        e^{-\Xi\ts{CA}} \notag \\
        &\times
        \omega_0( W( \Fk{A} \eta_\bk^*(t) ) )
        \omega_0( W( \Fk{B} \eta_\bk^*(t) ) ) \notag \\
        &\times 
        \omega_0( W( \Fk{C} \eta_\bk^*(t) ) )
        \, \\
    \rho_{23}
    &=
        e^{-2\ii t (\Delta\ts{B} - \Delta\ts{C} )}
        e^{ \ii (\vartheta\ts{AB} - \vartheta\ts{CA}) }
        e^{\Xi\ts{BC}} \notag \\
        &\times
        \omega_0( W( \Fk{B} \eta_\bk^*(t) ) )
        \omega_0( W( \Fk{C} \eta_\bk^*(t) ) )
        \,, \\
    \rho_{24}
    &=
        e^{-2\ii t (\Delta\ts{A} - \Delta\ts{C} )}
        e^{ \ii (\vartheta\ts{AB} - \vartheta\ts{BC}) }
        e^{\Xi\ts{CA}} \notag \\
        &\times
        \omega_0( W( \Fk{C} \eta_\bk^*(t) ) )
        \omega_0( W( \Fk{A} \eta_\bk^*(t) ) )
        \,, \\
    \rho_{25}
    &=
        e^{-2\ii t \Delta\ts{B} }
        e^{ \ii (\vartheta\ts{AB} - \vartheta\ts{BC}) }
        \omega_0( W( \Fk{B} \eta_\bk^*(t) ) )
        \,, \\
    \rho_{26}
    &=
        e^{-2\ii t \Delta\ts{A} }
        e^{ \ii (\vartheta\ts{AB} - \vartheta\ts{CA}) }
        \omega_0( W( \Fk{A} \eta_\bk^*(t) ) )
        \,, \\
    \rho_{27}
    &=
        e^{-2\ii t (\Delta\ts{A} + \Delta\ts{B} - \Delta\ts{C}) }
        e^{-\Xi\ts{AB}} e^{\Xi\ts{BC}} e^{\Xi\ts{CA}} \notag \\
        &\times 
        \omega_0( W( \Fk{A} \eta_\bk^*(t) ) )
        \omega_0( W( \Fk{B} \eta_\bk^*(t) ) ) \notag \\
        &\times 
        \omega_0( W( \Fk{C} \eta_\bk^*(t) ) )
        \,, \\
    \rho_{28}
    &=
        e^{-2\ii t (\Delta\ts{A} + \Delta\ts{B}) }
        e^{ \ii (-\vartheta\ts{CA} - \vartheta\ts{BC}) }
        e^{-\Xi\ts{AB}} \notag \\
        &\times 
        \omega_0( W( \Fk{A} \eta_\bk^*(t) ) )
        \omega_0( W( \Fk{B} \eta_\bk^*(t) ) )
        \,, \\
    \rho_{34}
    &=
        e^{-2\ii t (\Delta\ts{A} - \Delta\ts{B}) }
        e^{ \ii (\vartheta\ts{CA} - \vartheta\ts{BC}) }
        e^{\Xi\ts{AB}} \notag \\
        &\times 
        \omega_0( W( \Fk{A} \eta_\bk^*(t) ) )
        \omega_0( W( \Fk{B} \eta_\bk^*(t) ) )
        \,, \\
    \rho_{35}
    &=
        e^{-2\ii t \Delta\ts{C} }
        e^{ \ii (\vartheta\ts{CA} - \vartheta\ts{BC}) }
        \omega_0( W( \Fk{C} \eta_\bk^*(t) ) )
        \,, \\
    \rho_{36}
    &=
        e^{-2\ii t (\Delta\ts{A} - \Delta\ts{B} + \Delta\ts{C}) }
        e^{\Xi\ts{AB}} e^{\Xi\ts{BC}} e^{-\Xi\ts{CA}} \notag \\
        &\times 
        \omega_0( W( \Fk{A} \eta_\bk^*(t) ) )
        \omega_0( W( \Fk{B} \eta_\bk^*(t) ) ) \notag \\
        &\times 
        \omega_0( W( \Fk{C} \eta_\bk^*(t) ) )
        \,, \\
    \rho_{37}
    &=
        e^{-2\ii t \Delta\ts{A} }
        e^{ \ii (-\vartheta\ts{AB} + \vartheta\ts{CA}) }
        \omega_0( W( \Fk{A} \eta_\bk^*(t) ) )
        \,, \\
    \rho_{38}
    &=
        e^{-2\ii t (\Delta\ts{A} + \Delta\ts{C}) }
        e^{ \ii (-\vartheta\ts{AB} - \vartheta\ts{BC}) }
        e^{-\Xi\ts{CA}} \notag \\
        &\times 
        \omega_0( W( \Fk{C} \eta_\bk^*(t) ) )
        \omega_0( W( \Fk{A} \eta_\bk^*(t) ) )
        \,, \\
    \rho_{45}
    &=
        e^{2\ii t (\Delta\ts{A} - \Delta\ts{B} - \Delta\ts{C}) }
        e^{\Xi\ts{AB}} e^{-\Xi\ts{BC}} e^{\Xi\ts{CA}} \notag \\
        &\times 
        \omega_0( W( \Fk{A} \eta_\bk^*(t) ) )
        \omega_0( W( \Fk{B} \eta_\bk^*(t) ) ) \notag \\
        &\times 
        \omega_0( W( \Fk{C} \eta_\bk^*(t) ) )
        \,, \\
    \rho_{46}
    &=
        e^{-2\ii t \Delta\ts{C} }
        e^{ \ii (-\vartheta\ts{CA} + \vartheta\ts{BC}) }
        \omega_0( W( \Fk{C} \eta_\bk^*(t) ) )
        \,, \\
    \rho_{47}
    &=
        e^{-2\ii t \Delta\ts{B} }
        e^{ \ii (-\vartheta\ts{AB} + \vartheta\ts{BC}) }
        \omega_0( W( \Fk{B} \eta_\bk^*(t) ) )
        \,, \\
    \rho_{48}
    &=
        e^{-2\ii t (\Delta\ts{B} + \Delta\ts{C}) }
        e^{ \ii (-\vartheta\ts{AB} - \vartheta\ts{CA}) }
        e^{-\Xi\ts{BC}} \notag \\
        &\times 
        \omega_0( W( \Fk{B} \eta_\bk^*(t) ) )
        \omega_0( W( \Fk{C} \eta_\bk^*(t) ) )
        \,, \\
    \rho_{56}
    &=
        e^{-2\ii t (\Delta\ts{A} - \Delta\ts{B}) }
        e^{ \ii (-\vartheta\ts{CA} + \vartheta\ts{BC}) }
        e^{\Xi\ts{AB}} \notag \\
        &\times 
        \omega_0( W( \Fk{A} \eta_\bk^*(t) ) )
        \omega_0( W( \Fk{B} \eta_\bk^*(t) ) )
        \,, \\
    \rho_{57}
    &=
        e^{-2\ii t (\Delta\ts{A} - \Delta\ts{C}) }
        e^{ \ii (-\vartheta\ts{AB} + \vartheta\ts{BC}) }
        e^{\Xi\ts{CA}} \notag \\
        &\times 
        \omega_0( W( \Fk{C} \eta_\bk^*(t) ) )
        \omega_0( W( \Fk{A} \eta_\bk^*(t) ) )
        \,, \\
    \rho_{58}
    &=
        e^{-2\ii t \Delta\ts{A} }
        e^{ \ii (-\vartheta\ts{AB} - \vartheta\ts{CA}) }
        \omega_0( W( \Fk{A} \eta_\bk^*(t) ) )
        \,, \\
    \rho_{67}
    &=
        e^{-2\ii t (\Delta\ts{B} - \Delta\ts{C}) }
        e^{ \ii (-\vartheta\ts{AB} + \vartheta\ts{CA}) }
        e^{\Xi\ts{BC}} \notag \\
        &\times 
        \omega_0( W( \Fk{B} \eta_\bk^*(t) ) )
        \omega_0( W( \Fk{C} \eta_\bk^*(t) ) )
        \,, \\
    \rho_{68}
    &=
        e^{-2\ii t \Delta\ts{B} }
        e^{ \ii (-\vartheta\ts{AB} - \vartheta\ts{BC}) }
        \omega_0( W( \Fk{B} \eta_\bk^*(t) ) )
        \,, \\
    \rho_{78}
    &=
        e^{-2\ii t \Delta\ts{C} }
        e^{ \ii (-\vartheta\ts{CA} - \vartheta\ts{BC}) }
        \omega_0( W( \Fk{C} \eta_\bk^*(t) ) )
        \,.
\end{align*}
Here, $\vartheta_{ij}, \omega_0( W( F_\bk^{(j)} \eta_\bk^*(t) ) )$, and $\Xi_{ij}$ are given in Eq.~\eqref{eq:explicit forms for phases}, with appropriate coupling constants.

\bibliography{ref}
\end{document}